\documentclass[draftcls,onecolumn]{IEEEtran}
\usepackage{mathtools}
\usepackage{bbm,dsfont,mathrsfs,fixmath,amsthm}
\usepackage{amsmath,amssymb,eucal,graphicx}
\usepackage{stmaryrd,cite}
\usepackage{amsmath,amstext,amsfonts,amssymb}
\usepackage{epsfig}
\usepackage{exscale}
\usepackage{enumerate}
\usepackage{mathtools}
\usepackage{multirow}

\newtheorem{definition}{Definition}%[section]
\newtheorem{lemma}{Lemma}%[section]
\newtheorem{cor}{Corollary}%[section]
\newtheorem{theorem}{Theorem}
\newcommand{\cond}{\,\vert\,}

\newcommand{\defeq}{\triangleq}

\newfont{\bbb}{msbm10 scaled 500}

\newfont{\bb}{msbm10 scaled 1100}

\newcommand{\FF}{\mbox{\bb F}}

\newcommand{\DoF}{\mathsf{DoF}}

\newcommand{\Fc}{{\cal F}}

\newcommand{\Hc}{{\cal H}}
\newcommand{\Ic}{{\cal I}}
\newcommand{\Jc}{{\cal J}}
\newcommand{\Kc}{{\cal K}}
\newcommand{\Lc}{{\cal L}}
\newcommand{\Mc}{{\cal M}}

\newcommand{\Sc}{{\cal S}}

\newcommand{\Wc}{{\cal W}}

\renewcommand{\arg}{{\hbox{arg}}}

\newcommand{\eqdef}{\stackrel{\Delta}{=}}

\DeclareFontFamily{U}{cmfi}{}
\DeclareFontShape{U}{cmfi}{m}{n}{ <-> cmfi10 }{}
\DeclareSymbolFont{CMFI}{U}{cmfi}{m}{n}

\newcommand{\SNR}{\mathsf{snr}}

\makeatletter
\def\@copyrightspace{\relax}
\makeatother

\begin{document}

\title{Content Delivery in Erasure Broadcast Channels with Cache and Feedback}

\author{\IEEEauthorblockN{}
\IEEEauthorblockA{Asma Ghorbel, Mari Kobayashi, and Sheng Yang \\
LSS, CentraleSup\'elec \\
Gif-sur-Yvette, France\\
 {\tt \{asma.ghorbel, mari.kobayashi, sheng.yang\}@centralesupelec.fr}
}
}

\maketitle
\begin{abstract}
 We study a content delivery problem in a $K$-user erasure broadcast channel such
that a content providing server wishes to deliver requested files to users, each equipped with a
cache of a finite memory. 
 Assuming that the
 transmitter has state feedback and user caches can be filled during off-peak hours reliably by
 the decentralized content placement, we characterize the achievable rate region as a function of the
 memory sizes and the erasure probabilities.   
The proposed delivery scheme, based on the broadcasting
scheme by Wang and Gatzianas et al., exploits the
receiver side information established during the
placement phase. Our results can be extended to
centralized content placement as well as multi-antenna broadcast channels with state feedback. 
\end{abstract}

\section{Introduction}
Today's exponentially growing mobile data traffic is mainly due to video applications such as content-based video 
streaming. 
The skewness of the video traffic together with the ever-growing cheap on-board storage memory suggests
that the quality of experience can be boosted by 
caching popular contents at (or close to) the end-users in wireless networks. 
A number of recent works have studied such concept under different models and assumptions (see \cite{maddah2013fundamental,golrezaei2011femtocaching} and references therein). 
Most of existing works assume that caching is performed in two phases: {\it placement phase} to prefetch users' caches under their memory constraints 
(typically during off-peak hours) prior to the actual demands; {\it delivery phase} to transmit codewords 
such that each user, based on the received signal and the contents of its
cache, is able to decode the requested file. 
In this work, we study the delivery phase based on a coded caching model where a server is connected to many users, each equipped with a cache of finite memory \cite{maddah2013fundamental}. By carefully choosing the sub-files to be distributed
across users, coded caching exploits opportunistic multicasting such that a common signal is simultaneously
useful for all users even with distinct file requests. A number of extensions of coded caching have
been developed (see e.g. \cite[Section VIII]{maddah2013fundamental}). These include the decentralized
content placement \cite{maddah2013decentralized}, online coded caching \cite{pedarsani2013online},
non-uniform popularities \cite{ji2015order,niesen2013coded}, more general networks such as device-to-device (D2D) enabled network \cite{ji2013fundamental}, hierarchical networks \cite{Hierarchical2014}, heterogeneous networks \cite{hachem2015effect}, as well as the performance analysis in different regimes \cite{Allerton2014,Ajaykrishnan}. 
Further, very recent works have attempted to relax the unrealistic assumption of a perfect shared
link by replacing it by wireless channels (e.g.
\cite{huang2015performance,timo2015joint,zhang2015coded,allerton2015,Elia2015}). If wireless channels
are used only to multicast a common signal,  naturally the performance of coded caching (delivery
phase) is limited by the user in the worst condition of fading channels as observed in
\cite{huang2015performance}. This is due to the information theoretic limit, that is, the
multicasting rate is determined by the worst user \cite[Chapter 7.2]{el2011network}. However, if the
underlying wireless channels enjoy some degrees of freedom to convey simultaneously both private
messages and common messages, the delivery phase of coded caching can be further enhanced. 
In the context of multi-antenna broadcast channel and erasure broadcast channel, the potential gain of   
coded caching in the presence of channel state feedback has been demonstrated \cite{zhang2015coded,allerton2015,Elia2015}. The key observation behind \cite{allerton2015,Elia2015} is that opportunistic multicasting can be performed based on either the receiver side information established during the placement phase or the channel state information acquired via feedback. 

In this work, we model the bottleneck link between the server with $N$ files and $K$ users equipped
with a cache of a finite memory as an erasure broadcast channel (EBC). The simple EBC captures the essential features of wireless channels such as random failure or disconnection of any server-user link that a packet transmission may experience during high-traffic hours, i.e. during the delivery phase. 
In this work, we consider a memoryless EBC in which erasure is independent across users with
probabilities $\{\delta_k\}$ and each user $k$ can cache up to $M_k$ files.  
Moreover, the server is assumed to acquire the channel states causally via feedback sent by the
users. Assuming that users fill the caches randomly and independently according to the
decentralized content placement scheme as proposed in~\cite{maddah2013decentralized}, we study the achievable rate region of the EBC with cache and state feedback. 
Our main contribution is the characterization of the rate region in the cache-enabled EBC with state feedback for the case of the decentralized content placement (Theorem 1). 
The converse proof builds on the genie-aided bounds exploiting two key lemmas, i.e. a generalized
form of the entropy inequalities (Lemma 1) as well as the reduced entropy of messages in the presence
of receiver side information (Lemma 2). For the achievability, we present a multi-phase delivery scheme extending the
algorithm proposed independently by Wang~\cite{wang2012capacity} and by Gatzianas et
al.~\cite{gatzianas2013multiuser} to the case with receiver side information and prove that it
achieves the optimal rate region for special cases of interest. We provide, as a byproduct of the achievability proof for the symmetric network, an alternative proof for the sum capacity of the EBC with state feedback and without cache. More specifically, we characterize the order-$j$ capacity defined as the maximum transmission rate of a message intended to $j$ users and express the sum capacity in a convenient manner along the line of \cite{maddah2010degrees}. 
This allows us to characterize the rate region of the symmetric cache-enabled EBC with state
feedback easily, since as such all we need is to incorporate the packets generated during the
placement phase~\cite{allerton2015}. 
However, such proof exploits the specific structure of 
the rate region of symmetric networks, 
and unfortunately cannot be applied to a general network setting considered here. 
Our current work provides a non-trivial extension of \cite{allerton2015} to such networks. 
Furthermore, we show that our results can be extended in a straightforward manner to the centralized
content placement \cite{maddah2013fundamental} as well as the multi-antenna broadcast channel with
state feedback. Finally, we provide some numerical examples to quantify the benefit of state feedback, the relative merit of the centralized caching to the decentralized counterpart, as well as the gain due to the optimization of memory sizes, as a function of other system parameters.

The rest of the paper is organized as follows. In section~\ref{section:MainResult}, we describe the
system model together with some definitions and then summarize the main results. Section~\ref{section:UpperBound} gives the converse proof of the achievable rate region of the cache-enabled
EBC with state feedback. After a high-level description of the well-known algorithm by Wang and
Gatzianas et al. in section~\ref{section:Revisiting}, section~\ref{section:achievability2} presents our proposed delivery scheme and provides the achievability proof for some special cases of interest. Section~\ref{section:Extensions} provides the extensions of the previous results and section~\ref{section:Examples} shows some numerical examples.

Throughout the paper, we use the following notational conventions. The superscript notation $X^n$ represents a sequence $(X_1,\ldots,X_n)$  of variables.
$X_{\Ic}$ is used to denote the set of variables $\{X_i\}_{i\in\Ic}$. 
The entropy of $X$ is denoted by $H(X)$. We let $[k]=\{1,\dots, k\}$. We let $\epsilon_n$ denote a constant which vanishes as $n\rightarrow \infty$, i.e. 
$\lim_{n\rightarrow \infty} \epsilon_n=0$. 

%%%%%%%%%

\section{System Model and Main Results} \label{section:MainResult}
\subsection{System model and definitions}
\begin{figure}
\vspace{-10pt}
\begin{center}
\includegraphics[width=0.45\textwidth,clip=]{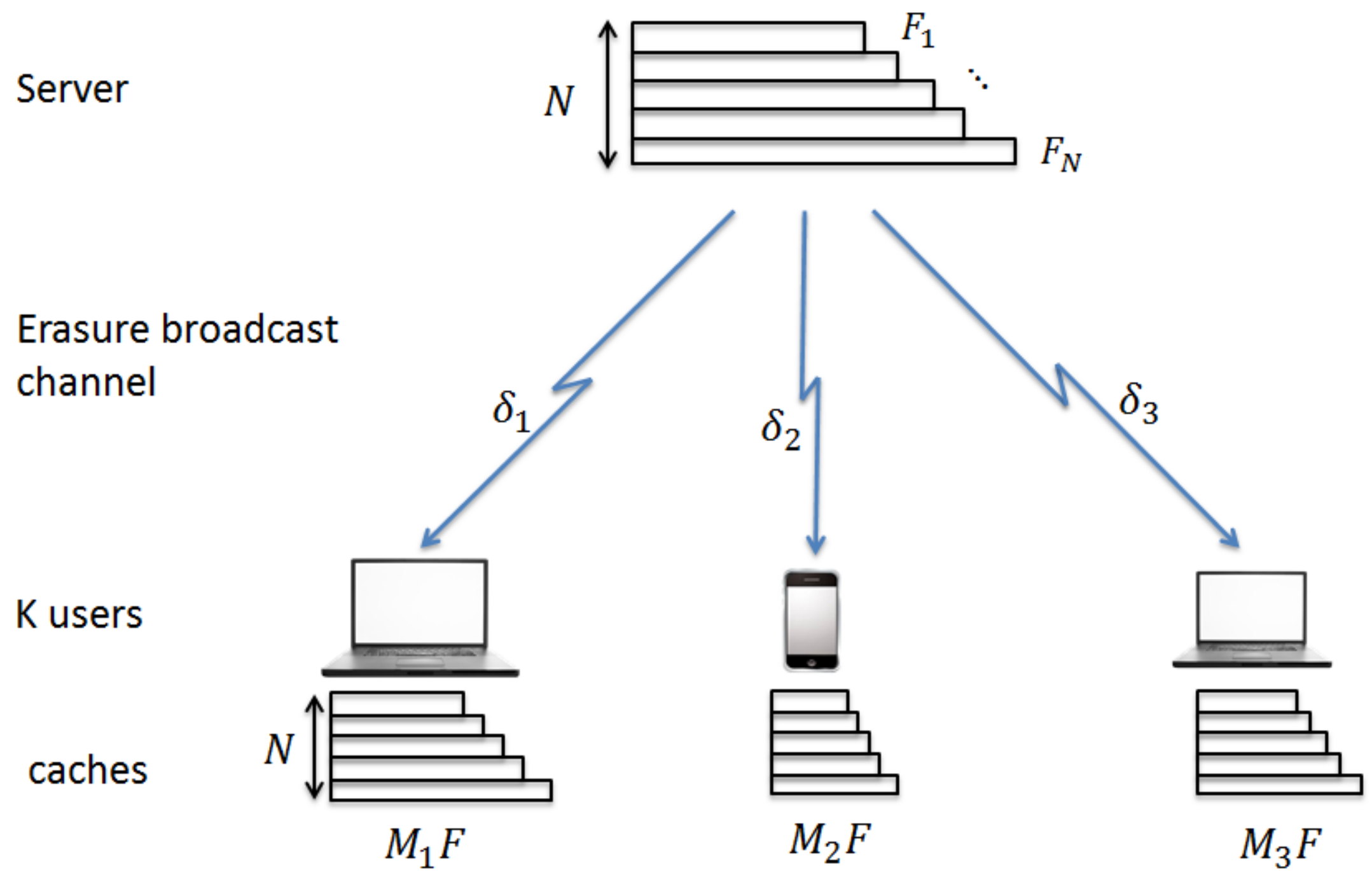}
\vspace{-2pt}
\caption{A cached-enabled erasure broadcast channel with $K=3$.}
\label{fig:model}
\end{center}
\vspace{-15pt}
\end{figure}
We consider a cache-enabled network depicted in Fig. \ref{fig:model} where a server is connected to $K$ users through an erasure broadcast channel (EBC). 
The server has an access to $N$ files $W_1, \dots, W_N$ where the $i$-th file $W_i$ consists of $F_i$ packets of $L$ bits each ($F_i L$ bits). Each user
$k$ has a cache memory $Z_k$ of $M_kF$ packets for $M_k\in [0,N]$, where
$F\defeq\frac{1}{N}\sum_{i=1}^N F_i$ is the average size of the files. 
Under such a setting, consider a discrete time communication system where a packet is sent in
each slot over the $K$-user EBC. The channel input $X_k\in  \FF_q$ belongs to the input alphabet
of size $L\defeq\log_2 q$ bits.
 %\footnote{We assume that $L> \log_2 K$ so that the achievability results of \cite{gatzianas2013multiuser} hold.}. 
 The erasure is assumed to be memoryless and independently distributed across users so that in a given slot we have 
\begin{align}\label{eq:EBC}
& \Pr( Y_{1}, Y_{2}, \dots, Y_{K}|X ) = \prod_{k=1}^K \Pr(Y_k | X) \\
& \Pr ( Y_k |X) =\begin{cases}
1-\delta_k,& Y_k=X,\\
\delta_k,& Y_k = E
\end{cases}
\end{align}
where $Y_k$ denotes the channel output of receiver $k$, $E$ stands for an erased output, $\delta_k$ denotes the erasure probability of user $k$. We let $S_l\in \Sc=2^{\{1,\dots, K\}}$ denote the
state of the channel in slot $l$ and indicate the set of users who received
correctly the packet. We assume that all the receivers know instantaneously $S_l$, and that
through feedback the transmitter only knows the past states $S^{l-1}$ during slot $l$.

The caching network is operated in two phases: the placement phase and the delivery phase.
In the content placement phase, the server fills the caches of all users, $Z_1, \dots,
Z_K$, up to the memory constraint. As in most works in the literature, we
assume that the placement phase incurs no error and no cost,
since it takes place usually during off-peak traffic hours. Once each user
$k$ makes a request $d_k$, the server sends the codewords so that each user can
decode its requested file as a function of its cache contents and received
signals during the delivery phase. We provide a more formal definition below.  A
$(M_1, \dots, M_K, F_{d_1}, \dots, F_{d_K}, n)$ caching scheme consists of the following
components.
\begin{itemize}
\item $N$ message files $W_1,\dots, W_N$ independently and uniformly distributed over $\Wc_1 \times \dots \times \Wc_N$ with $\Wc_i = \FF_q^{F_i}$ for all $i$.
\item $K$ caching functions defined by $\phi_k: \FF_q^{\sum_{i=1}^N F_i} \to\FF_q^{FM_k}$ that map
  the files $W_1, \dots, W_N$ into user~$k$'s cache contents 
\begin{align}
Z_k = \phi_k(W_1,\dots, W_N), \quad k\in[K].
\end{align}
\item A sequence of encoding functions which transmit at slot $l$ a symbol
  $X_l = f_l (W_{d_1}, \dots, W_{d_K}, S^{l-1})\in \FF_q$, based on the
  requested files and the state feedback up to slot $l-1$ for $l=1,
  \dots, n$, where $W_{d_k}$ denotes the message file requested by user $k$ for $d_k \in \{1, \dots, N\}$. 
\item $K$ decoding functions defined by $\psi_k: \FF_q^n
  \times\FF_q^{FM_k} \times \Sc^n \to \FF_q^{F_{d_k}}$, $k\in[K]$, that
decode the file $\hat{W}_{d_k}= \psi_k(Y_k^n, Z_k, S^n)$ as a function of the received signals $Y_k^n$, the cache content $Z_k$, as well as the state information $S^n$. 
\end{itemize}
A rate tuple $(R_1,\dots, R_K)$ is said to be achievable if, for every
$\epsilon>0$, there exists a $(M_1, \dots, M_K, F_{d_1},\dots, F_{d_K}, n)$ caching strategy that satisfies the reliability condition 
\[ \max_{(d_1,\dots, d_K)\in \{1,\dots, N\}^K} \max_k  \Pr( \psi_k(Y_k^n, Z_k, S^n) \neq W_{d_k} ) < \epsilon \]
as well as the rate condition 
\begin{align}
 R_k  <  \frac{F_{d_k}}{n}\quad \forall k\in[K]. 
\end{align} 
Throughout the paper, we express for brevity the entropy and the rate in terms of packets in oder to avoid the constant factor $L=\log_2 q$.

\subsection{Decentralized content placement}\label{subsection:Decentralized}
We mainly focus on the decentralized content placement proposed in \cite{maddah2013decentralized}. Under the memory constraint of $M_kF$ packets, each user $k$ independently caches a subset of $p_kF_i$  packets of file $i$, chosen uniformly at random for $i=1,\dots, N$, where $p_k=\frac{M_k}{N}$. 
By letting $\Lc_{\Jc}(W_i)$ denote the sub-file of $W_i$ stored exclusively by the users in $\Jc$, the cache memory of user $k$ after the decentralized placement is given by
\begin{align} \label{eq:Zk}
Z_k =\{ \Lc_{\Jc} (W_i):  \Jc \subseteq [K], \ \Jc \ni k, \ i =1,\dots,N \}. 
\end{align} 
The size of each sub-file is given by
\begin{align}\label{eq:subfileWk}
  |\Lc_{\Jc}(W_i)|=\prod_{j\in\Jc}p_j\prod_{j\in[K]\setminus\Jc}(1-p_{j}) F_i +
  \epsilon_{F_i}
\end{align}
as $F_i\rightarrow \infty$. It can be easily verified that the memory constraint of each user  
is fulfilled, namely,
\begin{align}
  |Z_{k}|=\sum_{i=1}^{N}\sum_{\Jc:k\in\Jc}|\Lc_{\Jc}(W_i)|= %=\sum_{i=1}^{N}\left(F_i \sum_{\Jc:k\in\Jc}\prod_{j\in\Jc}p_j\prod_{j\in[K]\setminus\Jc}(1-p_{j})+ \epsilon_{F_i}\right)=
  \sum_{i=1}^{N}(F_i p_k+\epsilon_{F_i}) = M_kF + \sum_{i=1}^{N} \epsilon_{F_i}\label{eq:LLN}
\end{align}%
as $F_i\rightarrow \infty$ for all $i$. 
Throughout the paper, we assume that $F\to\infty$ and meanwhile $\frac{F_i}{F}$ converges to
some constant $\tilde{F}_i>0$. Thus, we identify all $\epsilon_{F_i}$ with a single
$\epsilon_F$. 

To illustrate the placement strategy, let us consider an example of $K=3$ users. 
After the placement phase, each file will be partitioned into 8 sub-files:
\begin{align}\label{eq:DecentralizedPlacement}
W_i=\{\Lc_{\emptyset}(W_{i}),\Lc_{1}(W_{i}), \Lc_{2}(W_{i}), \Lc_{3}(W_{i}),\Lc_{12}(W_{i}), \Lc_{13}(W_{i}), \Lc_{23}(W_{i}), \Lc_{123}(W_{i}) \}. 
\end{align}
Obviously, the sub-files received by the destination, e.g. $\Lc_{1}(W_1), \Lc_{12}(W_1),
\Lc_{13}(W_1),\Lc_{123}(W_1) $ for user 1 requesting $W_1$, need not be transmitted in the delivery phase.

\subsection{Main results}\label{subsection:Main_Result}

In order to present the main results, we specify two special cases. 
\begin{definition}
The cache-enabled EBC (or the network) is \emph{symmetric} if the erasure probabilities as well as the memory sizes are the same for all users, i.e. $\delta_1=\dots=\delta_K=\delta$, $M_1=\dots =M_K=M$, $p_1=\dots =p_K=p$.
\end{definition}
\begin{definition}\label{definition:onesidedfair}
The rate vector is said to be \emph{one-sided fair} in the cache-enabled EBC if $\delta_k \geq \delta_j$ and for $k\neq j$ implies 
\begin {align}
 %\begin{cases}
 \frac{1-p_k}{p_k} R_k &\geq \frac{1-p_j}{p_j}R_j,\quad\text{and}\\
 \delta_k R_k &\geq \delta_jR_j. 
 %\end{cases}  
\end {align}	
\end{definition}
For the special case where $p_k=0$, $\forall k\in[K]$, it is reduced to $\delta_k R_k \geq
\delta_jR_j$ which coincides with the one-sided fairness originally defined in \cite{wang2012capacity}. 

\begin{figure}
	\vspace{-10pt}
	\begin{center}
		\includegraphics[width=0.6\textwidth,clip=]{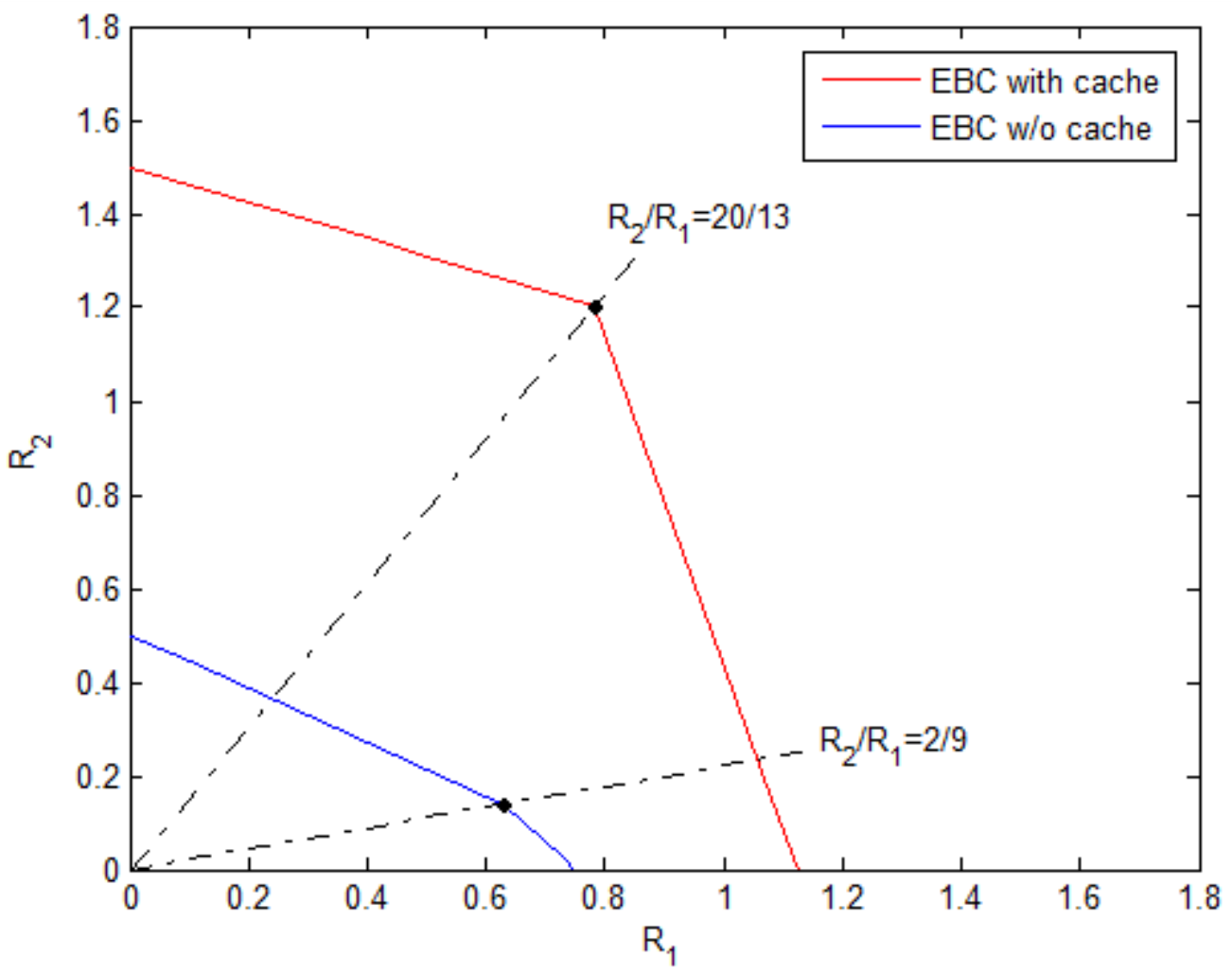}
		\vspace{-2pt}
		\caption{A two-user rate region with $(p_1, p_2)=(\frac{1}{3},\frac{2}{3})$, $(\delta_1,\delta_2)=(\frac{1}{4},\frac{1}{2})$.}
		\label{fig:capacity}
	\end{center}
	\vspace{-15pt}
\end{figure}

Focusing on the case of most interest with $N\geq K$ and $K$ distinct demands, we present the
following main results of this work.
\begin{theorem} \label{theorem:region}
For $K\leq 3$, or for the symmetric network with $K\geq3$, or for the one-sided fair rate vector with $K>3$, 
the achievable rate region of the cached-enabled EBC with the state feedback under the decentralized content placement
%, where user $k$ has a memory of size $p_kNF$ packets and an erasure probability $\delta_k$, 
is given by
\begin{align} \label{eq:wsr}
\sum_{k=1}^{K}\frac{\prod_{j=1}^{k}(1-p_{\pi_j})}{1-\prod_{j=1}^{k}\delta_{\pi_j}} R_{\pi_k}\leq 1
\end{align}
for any permutation $\pi$ of $\{1,\dots, K\}$. 
\end{theorem}
The above region has a polyhedron structure determined by $K!$ inequalities in general. It should be remarked that Theorem \ref{theorem:region} covers some existing results. 
 For the symmetric network, the above region simplifies to \cite{allerton2015}
\begin{align} \label{eq:symmetric}
\sum_{k=1}^{K}\frac{(1-p)^k}{1-\delta^k}R_{\pi_k}\leq 1, ~~\forall \pi.
\end{align}
For the case without cache memory, i.e. $p_k=0$ for all $k$,  
Theorem \ref{theorem:region} boils down to the capacity region of the EBC with state feedback \cite{wang2012capacity,gatzianas2013multiuser} given by
\begin{align} \label{eq:asymmetric}
\sum_{k=1}^{K}\frac{1}{1-\prod_{j=1}^k \delta_{\pi_j}} R_{\pi_k}\leq 1, ~~\forall \pi
\end{align}
which is achievable for $K\leq 3$ or the symmetric network or the one-sided fair rate vector where $\delta_k  \geq \delta_j$ implies $\delta_k R_k \geq \delta_jR_j$
for any $k\neq j$. 
Comparing \eqref{eq:wsr} and \eqref{eq:asymmetric}, we immediately see that the presence of
cache memories decreases the weights in the weighted rate sum and thus enlarges the rate region.
In order to gain some further insight, 
Fig. \ref{fig:capacity} illustrates a toy example of two users with
$(p_1,p_2)=(\frac{1}{3},\frac{2}{3})$ and $(\delta_1,\delta_2)=(\frac{1}{4},\frac{1}{2})$.
According to Theorem \ref{theorem:region}, the rate region is given by
\begin{align}
\frac{8}{9}R_1+\frac{16}{63}R_2\leq 1  \nonumber \\
\frac{16}{63}R_1+\frac{2}{3}R_2\leq 1 
\end{align}
which is characterized by three vertices $(\frac{9}{8},0)$ $(0.78, 1.20)$, and $(0,\frac{63}{16})$. The vertex $(0.78, 1.20)$, achieving the sum rate of $1.98$, corresponds to the case when the requested files satisfy the ratio $F_{d_2}/F_{d_1}= 20/13$. 
On the other hand, the region of the EBC without cache is given by 
\begin{align}
\frac{4}{3}R_1+\frac{8}{7}R_2\leq 1  \nonumber \\
\frac{8}{7}R_1+2R_2\leq 1
\end{align}
which is characterized by three vertices $(\frac{3}{4},0)$, $(0.63, 0.14)$, $(0,\frac{1}{2})$. The sum capacity of $0.77$ is achievable for the ratio 
$R_2/R_1=2/9$. The gain due to the cache is highlighted even in this toy example. 

Theorem \ref{theorem:region} yields the following corollary. 
 \begin{cor}\label{cor:rate}
For $K\leq 3$, or for the symmetric network with $K\geq3$, or for the one-sided fair rate vector with $K>3$, the transmission length to deliver requested files to useres in the cached-enabled EBC under the decentralized content placement is given by
\begin{align}\label{eq:time}
T_{\rm tot}= \max_{\pi}\left\lbrace \sum_{k=1}^{K}\frac{
\prod_{j=1}^{k}(1-p_{\pi_j})}{1-\prod_{j=1}^{k}\delta_{\pi_j}} F_{d_{\pi_k}} \right\rbrace
+ \epsilon_{F},
\end{align}
as $F\rightarrow \infty$. 
\end{cor}
The corollary covers some existing results in the literature. For the symmetric network with files of equal size ($F_i=F,\forall i$), the transmission length simplifies to 
\begin{align}\label{eq:rate2}
T_{\rm tot}=\sum_{k=1}^{K}\frac{ (1-p)^{k}}{1-\delta^{k}}F+\epsilon_{F},
\end{align}
as $F\rightarrow \infty$ \cite{allerton2015}. For the case with files of equal size and without
erasure, the transmission length in Corollary \ref{cor:rate} normalized by $F$ coincides with the ``rate-memory tradeoff'' \footnote{In
\cite{maddah2013decentralized} and all follow-up works, the ``rate'' is defined as the number of files
to deliver over the shared link, which corresponds to our $T_{\rm tot}$ here.} under the decentralized content
placement for asymmetric memory sizes \cite{wang2015fundamental} given by
\begin{align}\label{eq:Wang}
\frac{T_{\rm tot}}{F}=\sum_{k=1}^{K}\left[ \prod_{j=1}^{k}\left(1-p_j \right)\right], %\frac{M_j}{N}
\end{align}
where the maximum over all permutations is chosen to be identity by assuming $p_1\geq\cdots\geq p_K$.  %which implies that the maximum over all permutations in \eqref{eq:time} is given by identity.}
If additionally we restrict ourselves to the case with caches of equal size, we recover the
rate-memory tradeoff given in \cite{maddah2013decentralized} 
\begin{align}
\frac{T_{\rm tot}}{F}=  \frac{N}{M} \left(1-\frac{M}{N}\right) \left\{1-\left(1-\frac{M}{N}\right) ^K\right\}.
\end{align}
In fact, the above expression readily follows by applying the geometric series to the RHS of \eqref{eq:Wang}.   

%%%%%%%%%
\section{Converse} \label{section:UpperBound}
In this section, we prove the converse of Theorem \ref{theorem:region}. First we provide two useful lemmas. 
The first one is a generalized form of the entropy inequality, while the second one is a simple relation of the message entropy in the presence of receiver side information. Although the former has been proved in \cite{allerton2015}, we restate it for the sake of completeness. 
\begin{lemma}\cite[Lemma 5]{ShengISIT2015}\label{lemma:erasure-ineq}
  For the erasure broadcast channel, if $U$ is such that $X_l\leftrightarrow U  Y_{\Ic}^{l-1}  S^{l-1} \leftrightarrow
  (S_{l+1},\ldots,S_n)$, $\forall\,\Ic$,
 \begin{align} \label{eq:essential-erasure}
  \frac{1}{1-\prod_{j\in \Ic} \delta_j} H(Y^n_{\mathcal{I}} \cond U, S^n) \le \frac{1}{1-\prod_{j\in \Jc} \delta_j }
  H(Y^n_{\mathcal{J}} \cond U, S^n)
\end{align}  
for any sets $\Ic, \Jc$ such that $\Jc \subseteq \Ic \subseteq \left\{ 1,\ldots,K \right\}$.  
\end{lemma}
\begin{proof}
%\subsection{Proof of Lemma~\ref{lemma:erasure-ineq}} \label{appendix:erasure-ineq}
%\begin{IEEEproof}[Proof of Lemma~\ref{lemma:erasure-ineq}] %
  %Let us define $S_i$ as the set of indices of the receivers \emph{not} in erasure at time instant~$i$, i.e.,$S_i\defeq\{k\,:\ Y_{k,i} = X_{k,i} \}$. Then, 
  We have, for $\Jc\subseteq\Ic$,
  \begin{align}
    \MoveEqLeft[0]{H(Y^n_{\mathcal{I}} \cond U, S^n)}\\ 
    &= \sum_{l=1}^n H(Y_{\mathcal{I}, l} \cond
    Y^{l-1}_{\mathcal{I}}, U, S^n) \\
    &= \sum_{l=1}^n H(Y_{\mathcal{I},l} \cond Y^{l-1}_{\mathcal{I}}, U, S^{l-1}, S_l) \\
    &= \sum_{l=1}^n \mathrm{Pr}\{S_l\cap\mathcal{I}\ne\emptyset\} \, H(X_l \cond Y^{l-1}_{\mathcal{I}}, U, S^{l-1}, S_l\cap\mathcal{I}\ne\emptyset) \\
    &= \sum_{l=1}^n \bigl(1-\prod_{i\in\Ic}\delta_i\bigr) H(X_l \cond Y^{l-1}_{\mathcal{I}}, U, S^{l-1}) \\
    &\le \bigl(1-\prod_{i\in\Ic}\delta_i\bigr) \sum_{l=1}^n  H(X_l \cond Y^{l-1}_{\mathcal{J}}, U, S^{l-1})
    \label{eq:tmp821}
  \end{align}%
  where the first equality is from the chain rule; the second equality is because the current input does not
  depend on future states conditioned on the past outputs/states and $U$; the third one holds since $Y_{\mathcal{I},l}$ is deterministic and has entropy $0$
  when all outputs in $\mathcal{I}$ are erased~($S_l\cap\mathcal{I}=\emptyset$); the fourth equality is from
  the independence between $X_l$ and $S_l$; and we get the last inequality by removing the terms
  $Y^{l-1}_{\Ic\setminus\Jc}$ in the condition of
  the entropy. Following the same steps, we have 
  \begin{align}
    \MoveEqLeft[0]{H(Y^n_{\mathcal{J}} \cond U, S^n) = 
    \bigl(1-\prod_{i\in\Jc}\delta_i\bigr) \sum_{l=1}^n  H(X_l \cond Y^{l-1}_{\mathcal{J}}, U, S^{l-1})},
    %\label{eq:tmp822}
  \end{align}%
  from which and \eqref{eq:tmp821}, we obtain \eqref{eq:essential-erasure}.
\end{proof}

\begin{lemma} \label{lemma:decentralized}
Under the decentralized content placement
\cite{maddah2013decentralized}, the following
inequality hold for any $i$ and $\Kc\subseteq[K]$  
\[ H(W_i \cond \{Z_{k}\}_{k\in \Kc}) \ge
\prod_{k\in\Kc}\left(1-p_k\right) H(W_i).\]
 %as $F_i\rightarrow \infty$.
\end{lemma}

\begin{proof}
Under the decentralized content placement, we have
\begin{align}
 H(W_i \cond \{Z_{k}\}_{k\in \Kc})& = H(W_i \cond  \{\Lc_{\Jc}(W_l) \}_{  \Jc\cap \Kc\ne \emptyset,\, l=1,\dots, N})\label{eq:tmp8921}\\
  & = H(W_i \cond  \{ \Lc_{\Jc}(W_i) \}_{  \Jc\cap \Kc\ne \emptyset} )\\
  & = H( \{ \Lc_{\Jc}(W_i)\}_{  \Jc\cap \Kc = \emptyset} ) \\
  & = \sum_{\Jc:\Jc\subseteq [K]\setminus\Kc} H( \Lc_{\Jc}(W_i) ) \label{eq:tmp892}\\
  & \ge \sum_{\Jc:\Jc\subseteq [K]\setminus\Kc} H( \Lc_{\Jc}(W_i) \cond \Lc_{\Jc}) 
\end{align}
where the first equality follows from \eqref{eq:Zk}; the second equality follows due to the independence between messages $W_1,\cdots,W_N$; the third equality follows by identifying the unknown parts of $W_i$ given
the cache memories of $\Kc$ and using the
independence of all sub-files; \eqref{eq:tmp892} is
again from the independence of the sub-files. Note
that $\Lc_\Jc$ is a random variable indicating which
subset of packets of file $W_i$ are shared by the
users in $\Jc$. The size of the random subset
$|\Lc_\Jc|$ follows thus the binomial distribution
$B\Bigl(H(W_i),\,
\prod_{j\in\Jc}p_j\prod_{k\in[K]\setminus\Jc}(1-p_k)\Bigr)$.
It is readily shown that $H( \Lc_{\Jc}(W_i) \cond
\Lc_{\Jc}) = \mathbb{E}\{|\Lc_{\Jc}|\}$. This
implies that  
\begin{align}
{H(W_i \cond \{Z_{k}\}_{k\in \Kc} )}
  &\ge \sum_{\Jc:\Jc\subseteq [K]\setminus\Kc} \prod_{j\in\Jc}p_j\prod_{k\in[K]\setminus\Jc}(1-p_k)H( W_{i} ) \label{eq:tmp893} \\
  & = \prod_{k\in\Kc}(1-p_k)\sum_{\Jc:\Jc\subseteq [K]\setminus\Kc } \prod_{j\in\Jc}p_j\prod_{k\in[K]\setminus\Kc\setminus\Jc}(1-p_k)H( W_{i} )  \\
  & = \prod_{k\in\Kc}\left(1-p_k\right)H(W_{i}) 
\end{align}
where the last inequality is obtained from the basic property that we have
$\sum_{\Jc\subseteq\Mc}\prod_{j\in\Jc}p_j\prod_{k\in\Mc\setminus\Jc}(1-p_k)=1$
for a subset $\Mc=[K]\setminus\Kc$. %\\
\end{proof} 

We apply genie-aided bounds to create a degraded erasure broadcast channel by
providing the messages, the channel outputs, as well as the receiver side
information (contents of cache memories) to the enhanced receivers. Without loss of generality, we focus on
the case without permutation and the demand $(d_1,\dots, d_K)=(1,\dots, K$). %Considering decentralized cache placement, we have for user $k$, 
\begin{align}
n\prod_{j=1}^k (1-p_j)R_k &= \prod_{j=1}^k(1-p_j) H(W_k) \\
 &\le H(W_k |Z^k S^n) \\
 &\leq  I(W_k;Y_{[k]}^n \cond Z^k S^n) + n \epsilon'_{n,k} \label{eq:tmp722}\\
 &\leq  I(W_k;Y_{[k]}^n, W^{k-1} \cond Z^k S^n) + n \epsilon'_{n,k} \\
 &=  I(W_k;Y_{[k]}^n \cond  W^{k-1} Z^k S^n) + n \epsilon'_{n,k} 
\end{align}
where the second inequality is by applying Lemma \ref{lemma:decentralized} and
noting that $S^n$ is independent of others; \eqref{eq:tmp722} is from 
Fano's inequality; the last equality is from $I(W_k; W^{k-1} \cond
Z^k S^n) = 0$ since the caches $Z^k$ only store
disjoint pieces of individual files by the
decentralized content placement \cite{maddah2013decentralized}. Putting all the
rate constraints together, and defining
$\epsilon_{n,k} \defeq
\epsilon'_{n,k}/\prod_{j=1}^k(1-p_j)$, we have
\begin{align}
  n (1-p_1)(R_1 - \epsilon_{n,1}) &\leq   H(Y^n_1 \cond Z_1 S^n) - H(Y^n_1\cond
  W_1  Z_1 S^n) \nonumber \\
%n (R_2 - \epsilon_n) &\le I(W_2;Y^n_1 Y^n_2 \cond W_1 Z_1 Z_2 S^n) \\
%&= H(Y^n_1Y^n_2 \cond W_1 Z_1 Z_2) - H(Y^n_1\cond W_1 W_2  Z_1 Z_2S^n)\\
&\ \, \vdots \nonumber \\
n \prod_{j=1}^K(1-p_j) (R_K - \epsilon_{n,K}) &\le H(Y^n_{[K]} \cond W^{K-1} Z^{K}  S^n) - H(Y^n_{[K]}\cond W^{K} Z^{K} S^n).%\nonumber \\
%&\qquad - H(Y^n_{[K]}\cond W^{K} Z^{K} S^n).
\end{align}%
We now sum up the above inequalities with different weights, and apply $K-1$ times
Lemma~\ref{lemma:erasure-ineq}, namely, for $k=1,\ldots,K-1$,
\begin{align}
  \frac{H(Y^n_{[k+1]} \cond W^{k} Z^{k+1} S^n) }{1-\prod_{j\in[k+1]}\delta_{j}}  &\le  \frac{H(Y^n_{[k+1]} \cond W^{k} Z^{k} S^n)}{1-\prod_{j\in [k+1]}\delta_{j}}  \\ 
  &\le \frac{H(Y^n_{[k]} \cond W^{k} Z^{k} S^n)}{1-\prod_{j\in[k]} \delta_{j}}, %, \;\; \forall k=1,\dots, K-1
\end{align}%
where the first inequality follows because removing conditioning increases entropy. Finally, we have 
\begin{align}
%\Longrightarrow \\
\MoveEqLeft{\sum_{k=1}^{K} \frac{\prod_{j\in[k]}(1-p_j)}{1-\prod_{j\in[k]}\delta_{j}} (R_k - \epsilon_n)}
\nonumber \\
&\le \frac{H(Y^n_1  \cond Z_1 S^n)}{n(1-\delta_1)} - \frac{H(Y^n_{[K]} \cond W^{K} Z^K S^n)}{n(1-\prod_{j\in[k]}\delta_{j})}
\\
&\leq \frac{H(Y^n_1) }{n(1-\delta_1)}\le1
\end{align}%
which establishes the converse proof. 

%%%%%%%%%
\section{Broadcasting without receiver side information} \label{section:Revisiting}
In this section, we first revisit the algorithm 
proposed in
\cite{wang2012capacity,gatzianas2013multiuser} achieving the capacity region of the EBC
with state feedback for some cases of interest, as an
important building block of our proposed scheme. Then, we provide an alternative achievability proof for the symmetric channel with uniform erasure probabilities across users.

\subsection{Revisiting the algorithm by Wang and Gatzianas et al. }

We recall the capacity region of the EBC with state
feedback as below.
\begin{theorem}[\!\cite{wang2012capacity,gatzianas2013multiuser}] \label{theorem:EBC}
For $K\leq 3$, or for the symmetric channel with $K\geq3$, or for the one-sided fair rate vector\footnote{$\delta_k  \geq \delta_j$ implies $\delta_k R_k \geq \delta_jR_j$
for any $k\neq j$.} with $K>3$, the capacity region of the erasure broadcast channel with state feedback is given by
\begin{align} \label{eq:asymmetric2}
\sum_{k=1}^{K}\frac{1}{1-\prod_{j=1}^k \delta_{\pi_j}} R_{\pi_k}\leq 1, ~~\forall \pi.
\end{align}
\end{theorem}

We provide a high-level description of the broadcasting scheme \cite{wang2012capacity,gatzianas2013multiuser} which is optimal under the special cases as specified in the above theorem. 
%along the line of \cite{maddah2010degrees}  
We recall that the number of private packets $\{F_{k}\}$ is assumed to be arbitrarily large so
that the length of each phase becomes deterministic. Thus, we drop the $\epsilon_F$ term
wherever confusion is not probable. The broadcasting algorithm has two main roles: ~1) broadcast
new information packets and ~2)~multicast side information or overheard packets based on state feedback. 
Therefore, we can call phase 1 {\it broadcasting phase} and phases $2$ to $K$ {\it multicasting phase}. Phase $j$ consists of $K \choose j$ sub-phases in each of which the transmitter sends packets intended to a subset of users $\Jc$ for $|\Jc|=j$. Similarly to
the receiver side information obtained after the placement phase, we let $\Lc_{\Jc} (V_{\Kc})$
denote the part of packet $V_{\Kc}$ received by users in $\Jc$ and erased at users in $[K]\setminus \Jc$. 
%In the following we define a set of variables for $k\in\Ic\subseteq\Jc$:

Here is a high-level description of the broadcasting algorithm:
\begin{enumerate}
	\item Broadcasting phase (phase $1$): send each message $V_k= W_k$ of $F_k$ packets sequentially for $k=1,\dots, K$. 
	This phase generates overheard symbols $\{\Lc_{\Jc} (V_k)\}$ to be transmitted via linear combination in multicasting phase, where $ \Jc \subseteq [K] \setminus k$ for all $k$. 
	\item Multicasting phase (phases $2-K$): for a subset $\Jc$ of users, generate $V_{\Jc}$ as a linear combination of overheard packets such that
	\begin{align}%\label{eq:V_J}
		V_{\Jc} = \Fc_{\Jc}\left(\{\Lc_{\Jc \setminus \Ic \cup \Ic'} (V_{\Ic})\} _{\Ic' \Ic :  \Ic' \subset \Ic \subset \Jc} \right),
	\end{align}
	where $\Fc_{\Jc}$ denotes a linear function. Send $V_{\Jc}$ sequentially for all $\Jc \subseteq [K]$ of the cardinality $|\Jc|=2,\dots, K$.
\end{enumerate}
 
The achievability result of Theorem \ref{theorem:EBC} implies the following corollary.
\begin{cor}\label{lemma:GGT}
For $K\leq 3$, or for the symmetric channel with $K>3$, or for the one-sided fair rate vector
with $K>3$,  the total transmission length to convey $W_1,\dots, W_K$ to users $1,\dots, K$,
respectively, is given by   
  \[ T_{\rm tot}=\sum_{k=1}^K \frac{F_{\pi_k} }{1-\prod_{j=1}^k \delta_{\pi_j}}+\epsilon_{F}. \]
  %as $F_k\rightarrow \infty$ for all $k$. 
\end{cor}
The proof is omitted because the proof in section \ref{subsection:proof1} covers the case without user memories. 

\begin{table}[ht]
\caption{Notations for the erasure broadcast channel.}
\label{tab:1}
\begin{center}
\begin{tabular}{cc}\hline
$ R_k$ &  Message rate for user $ k $\\
$t_{\Jc}$ & Length of sub-phase $\Jc$\\
$ t_{\Jc}^{\{k\}} $ & Length needed by user $k$ for sub-phase $\Jc$ \\
$V_{\Kc}$ & Packets intended to users in $\Kc $ \\
$\Lc_{\Jc} (V_{\Kc})$ & Part of packets $V_{\Kc}$ received by users in $\Jc$ and erased at users
in $[K]\setminus \Jc$\\
$ N_{\Ic\rightarrow \Jc}^{\{k\}}$ & Number of packets useful for user $k$ generated in sub-phase $\Ic$ and to be sent in sub-phase $\Jc$ \\		
\hline
\end{tabular}
\end{center}
\end{table}
In order to calculate the total transmission length of the algorithm, we need to introduce
further some notations and parameters~(Table~\ref{tab:1}) which are explained as follows. 
\begin{itemize}
\item A packet intended to $\Jc$ is consumed {\it for a given user $k\in\Jc$}, if this user or at least one user in $[K]\setminus \Jc$ receives it. The probability
of such event is equal to $1- \prod_{j\in [K]\setminus\Jc\cup\{k\} } \delta_j$.
\item A packet intended to $\Ic$ becomes a packet intended to $\Jc$ and useful for user
  $k\in\Ic\subset\Jc\subseteq[K]$, if erased at user $k$ and all users in $[K]\setminus \Jc$ but received by $\Jc \setminus \Ic$. %The probability of such event is denoted by $\alpha^{\{k\}}_{\Ic \rightarrow \Jc} = \prod_{j'\in  [K]\setminus\Jc\cup\{k\}}\delta_{j'} \prod_{j \in \Jc \setminus \Ic}  (1-\delta_j)$.
The number of packets useful for user $k$ generated in sub-phase $\Ic$ and to be sent in sub-phase $\Jc$, denoted by $N^{\{k\}}_{\Ic\rightarrow \Jc}$, is then given by
\begin{align}\label{eq:Nij}
N^{\{k\}}_{\Ic\rightarrow \Jc}=t_{\Ic}^{\{k\}} \prod_{j'\in  [K]\setminus\Jc\cup\{k\}}\delta_{j'} \prod_{j \in \Jc \setminus \Ic}  (1-\delta_j) %\alpha^{\{k\}}_{\Ic\rightarrow \Jc},
%\frac{N^{\{k\}}_{\Ic}}{\beta^{\{k\}}_{\Ic}} \alpha^{\{k\}}_{\Ic\rightarrow \Jc},
\end{align}
where $t_{\Ic}^{\{k\}} $ denotes the length of sub-phase $\Ic$ viewed by user $k$ to be defined shortly. We can also express $N^{\{k\}}_{\Ic\rightarrow \Jc}$ as 
 \begin{align}\label{eq:NkIJ}
 N^{\{k\}}_{\Ic\rightarrow \Jc} = \sum_{\Ic'\subseteq \Ic\setminus k} |\Lc_{\Jc \setminus \Ic \cup \Ic'} (V^{\{k\}}_{\Ic})|,
 \end{align}
 where we let $V^{\{k\}}_{\Ic}$ denotes the part of $V_{\Ic}$ required for user $k$.
\item The duration $t_{\Jc}$ of sub-phase $\Jc$ is given by 
\begin{align}\label{eq:t_J}
t_{\Jc}=\max_{k\in\Jc}t_{\Jc}^{\{k\}} ,
\end{align}
 where 
\begin{align}\label{eq:t^k_J}
t_{\Jc}^{\{k\}}=\frac{\sum_{k\in\Ic\subset \Jc}  N^{\{k\}}_{\Ic\rightarrow \Jc}}{1- \prod_{j\in [K]\setminus\Jc\cup\{k\} } \delta_j}.
%\frac{N^{\{k\}}_{\Jc}}{\beta^{\{k\}}_{\Jc}}.
\end{align}
 \end{itemize}
The total transmission length is given by summing up all sub-phases, i.e. $T_{\rm tot} = \sum_{\Jc\subseteq[K]}  t_{\Jc}$. 

\begin{figure}
\vspace{-10pt}
\begin{center}
\includegraphics[width=0.45\textwidth,clip=]{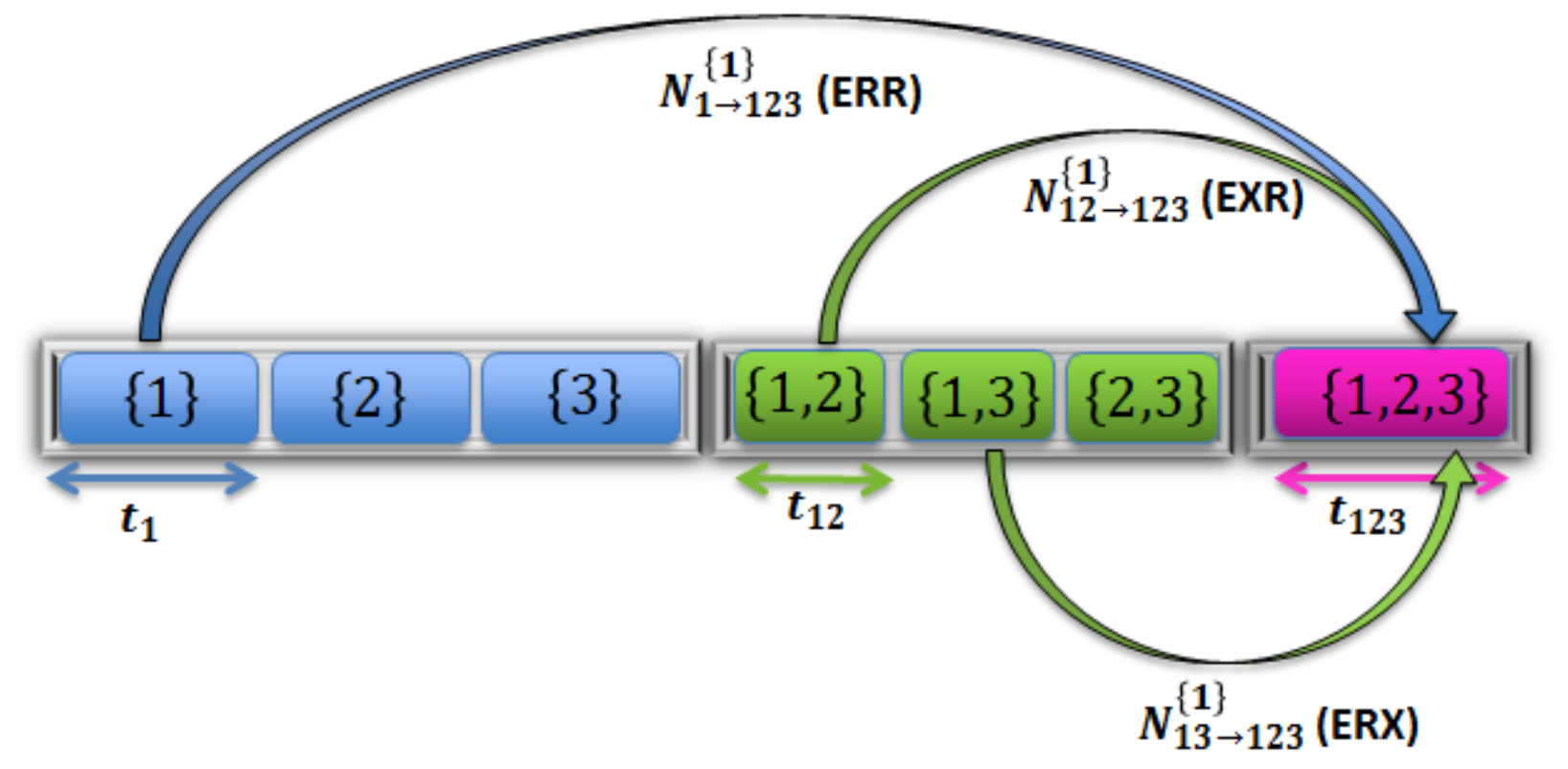}
\vspace{-2pt}
\caption{Phase organization for $K=3$ and packet evolution viewed by user 1.}
\label{fig:phase3}
\end{center}
\vspace{-10pt}
\end{figure}

Fig. \ref{fig:phase3} illustrates the phase organization for $K=3$ and the packet evolution viewed by user 1. The packets intended to $\{1,2,3\}$ are created from both phases 1 and 2. 
More precisely, sub-phase $\{1\}$ creates $\Lc_{23}(V_1)$ to be sent in phase 3 if erased at
user 1 and received by others (ERR). The number of such packets is $N^{\{1\}}_{1\rightarrow
123}$. Sub-phase $\{1,2\}$ creates $\Lc_{3}(V_{12}), \Lc_{23}(V_{12})$ if erased at user 1 but received by user 3 (EXR), while sub-phase $\{1,3\}$ creates $\Lc_{2}(V_{13}), \Lc_{23}(V_{13})$ if 
erased at user 1 and received by user 2 (ERX). 
The total number of packets intended to $\{1,2,3\}$ generated in phase 2 and required by user $1$ is 
%$|\Lc_{3}(V_{12})| + |\Lc_{13}(V_{12})|+|\Lc_{2}(V_{13})|+ |\Lc_{23}(V_{13})|$. 
$N^{\{1\}}_{12\rightarrow 123}+N^{\{1\}}_{13\rightarrow 123}$.

\subsection{Achievability in the symmetric channel}
We focus now on the special case of the symmetric channel with uniform erasure probabilities, i.e. $\delta_k=\delta$ for all $k$. In this case, the capacity region of the EBC with state feedback in \eqref{eq:asymmetric2} simplifies to 
\begin{align}\label{eq:regionEBC}
\sum_{k=1}^K \frac{1}{1-\delta^k} R_{\pi_k}\leq 1, ~~\forall \pi.
\end{align}
It readily follows that the capacity region yields the symmetric capacity, i.e. $R_1=\dots=R_K=R_{\rm sym}(K)$, given by 
\begin{align}
R_{\rm sym}(K) = \frac{1}{\sum_{k=1}^K \frac{1}{1-\delta^k}}.
\end{align} 
In the following, we provide an alternative proof of the achievability of the symmetric capacity.
Notice that other vertices of the capacity region can be characterized similarly as proved in subsection \ref{subsection:proof2}.
Our proof follows the footsteps of \cite{maddah2010degrees} and uses the notion of order-$j$ packets. Let us define message set $\{W_{\Jc}\}$
independently and uniformly distributed over $\{\Wc_{\Jc}\} $ for all $\Jc\subseteq [K]$. For $\Jc$ with the cardinality $j=|\Jc|$, the message set $\{W_{\Jc}\}$ are called order-$j$ messages. We define $R_{\Jc}$ an achievable rate of the message $W_{\Jc}$ and define the sum rate of order-$j$ messages as
\begin{align}
R^{j}(K) \defeq \sum_{\Jc:|\Jc|=j}R_{\Jc}={K \choose j}R_{\Jc}.
\end{align}
The supremum of $R^j(K)$ is called the sum capacity of order-$j$ messages. We characterize the sum capacity of order-$j$ messages,  in the erasure broadcast channel with state feedback in the following theorem. 

\begin{theorem} \label{theorem:order}
	In the $K$-user erasure broadcast channel with state feedback, the sum capacity of order-$j$ packets is upper bounded by
	\begin{align}\label{eq:Rorder}
	R^j (K) &\leq   \frac{{K \choose j}}{\sum_{k=1}^{K-j+1} \frac{ {K-k \choose j-1}}{1-\delta^k}} ,\;\; j=1,\dots, K.
	\end{align}
The algorithms in \cite{wang2012capacity,gatzianas2013multiuser} achieve the RHS with equality. 
\end{theorem}

\begin{proof}
We first provide the converse proof. Similarly to section \ref{section:UpperBound}, we build on genie-aided bounds together with Lemma \ref{lemma:erasure-ineq}. 
Let us assume that the transmitter wishes to convey the message $W_{\Jc}$ to a subset of users $\Jc \subseteq \{1,\dots, K\}$, and receiver $k$ wishes to decode all messages $\tilde{W}_k\eqdef \{W_{\Jc}\}_{\Jc: \Jc \ni k}$ for $j=1, \dots, K$. 
In order to create a degraded broadcast channel, we assume that receiver $k$ provides the message set $\tilde{W}_k$ and the channel output $Y_k^{n}$ to receivers $k+1$ to $K$ for $k=1,\dots, K-1$, Under this setting and using Fano's inequality, we have for receiver 1 :
\begin{align}
n \left(\sum_{1 \in \Jc \subseteq [K] }R_{\Jc}-\epsilon_{n,1} \right) &\leq  H(Y_1^n|S^n) -
H(Y_1^n |\tilde{W}_1S^n). 
\end{align}
For receiver $k=2, \dots, K$, we have: 
\begin{align}
n \left( \sum_{k \in \Jc \subseteq \{k,\dots, K\}} R_{\Jc}-\epsilon_{n,k} \right)  \leq  H(Y_1^n
\dots Y_k^n|\tilde{W}^{k-1} S^n)- H(Y_1^n \dots Y_k^n|\tilde{W}^{k} S^n),
\end{align}
where we used $\tilde{W}_k\setminus  \tilde{W}^{k-1}  = \{W_{\Jc}\}_{\Jc:\Jc \setminus \{k,\dots, K\}}$ in the LHS.
Summing up the above inequalities and applying Lemma \ref{lemma:erasure-ineq} $K-1$ times, we readily obtain:
\begin{align}\label{eq:upperbound}
\sum_{k=1}^K \frac{\sum_{k \in \Jc \subseteq \{k,\dots, K\}} (R_{\Jc}-\epsilon_{n,k})}{1-\delta^k} &\leq \frac{H(Y_1^n|S^n)}{n(1-\delta)} \\
& \leq 1.
\end{align}
We further impose the symmetric rate condition such that $R_{\Jc} = R_{\Jc'}$ for any $\Jc \neq \Jc'$ with the same cardinality. By focusing on $\Jc$ of the same cardinality $j$ in \eqref{eq:upperbound} and noticing that there are ${K-k \choose j-1}$ such subset,  $R_{\Jc}$ is upper bounded by  
\begin{align}
R_{\Jc}\leq \frac{1}{\sum_{k=1}^{K-j+1} \frac{ {K-k \choose j-1}}{1-\delta^k}}, \quad \forall\,\Jc,~|\Jc|=j. 
\end{align}
This establishes the converse part. \\

In order to prove the achievability of $R^i(K)$ in Theorem \ref{theorem:order}, we apply the broadcasting algorithm of~\cite{wang2012capacity,gatzianas2013multiuser} from phase $i>1$ by sending $N_i$ packets to each subset 
$\Ic \subseteq [K]$ with $|\Ic|=i$. First, we redefine some parameters by taking into account the symmetry across users as summarized in Table \ref{tab:2}.
Due to the symmetry, we drop the user index $k$ in $ t_{\Jc}^{\{k\}} $, $ N_{\Ic\rightarrow \Jc}^{\{k\}}$ and replace them by $t_j$, $N_{i\rightarrow j}$, respectively for $\Ic\subset \Jc \subseteq [K]$ with $|\Ic|=i, |\Jc|=j$. Now, we introduce 
variants of these notations to reflect the fact that the algorithm starts from phase $i>1$, rather than from phase 1. 
The length of any sub-phase in phase $j$ when starting the algorithm from phase $i$, denoted by $t^i_j$, is given by 
\begin{align}\label{eq:t^i_j}
t^i_j = \frac{1}{1-\delta^{K-j+1}} \sum_{l=i}^{j-1} {j-1 \choose l-1}N^i_{l\rightarrow j}, \;\; j> i,
\end{align}
where 
\begin{align}\label{eq:lengthN^i}
N^i_{l\rightarrow j} = t_l^i \delta^{K-j+1}(1-\delta)^{j-l}
\end{align}
 denotes the number of order-$j$ packets generated during a given sub-phase in phase $i$, again starting from phase $i$. 
 \begin{table}[t]
\caption{Notations for the symmetric channel.}
\label{tab:2}
\begin{center}
\begin{tabular}{cc}\hline
$W_{\Jc}$ & Message intended to users in $\Jc$ \\ 
$ R_{\Jc} $ & Rate of $W_{\Jc}$ \\ 
$ R^{j}(K)$ & Sum rate of order-$j$ messages  \\ 
%$R_{\rm sym}(K)$ & the symmetrical rate \\	\hline
$t_j = t_j^1$ & Length of any sub-phase in phase $j$ \\  
%$N_{i\rightarrow j}$ & \\ \hline
 %$\alpha_{i\rightarrow j}$
$ t_{j}^{i} $ & Length of any sub-phase $j$ when starting from phase $i$ \\ 
$ N_{i\rightarrow j}=N^{1}_{i\rightarrow j}$ & Number of packets created in sub-phase $\Ic$ and to be sent in sub-phase $\Jc$ for any $\Ic\subset \Jc$ of cardinality $i<j$ \\		
 $N^{i'}_{i\rightarrow j}$ & $N_{i\rightarrow j}$ when starting from phase $i'$ for $i\leq i' \leq j$\\ \hline
\end{tabular}
\end{center}	
\end{table}
 
 For $j=i$, we have 
\begin{align}\label{eq:t^i_i}
t^i_i = \frac{N_i}{1-\delta^{K-i+1}}.
\end{align}
By counting the total number of order-$i$ packets and the transmission length from phase $i$ to phase $K$, the sum rate of order-$i$ messages achieved by the algorithm \cite{wang2012capacity,gatzianas2013multiuser} is given by  
\begin{align}\label{eq:Rggl}
%R^{i}_{\rm GGT}(K)=\frac{{K \choose i}N_i}{\sum_{j=i}^{K}{K \choose j}t^{i}_j (N_i)}~~~~\forall i.
\tilde{R}^{i}(K)=\frac{{K \choose i}N_i}{\sum_{j=i}^{K}{K \choose j}t^{i}_j }, \quad\forall\,i.
\end{align} 
It remains to prove that $\tilde{R}^{i}(K)$ coincides with the RHS expression of \eqref{eq:Rorder}. 
We notice that the transmission length from phase $j$ to $K$ can be expressed in the following different way, i.e. 
\begin{align}\label{eq:e65}
\sum_{j=i}^{K}{K \choose j}t^{i}_j= \sum_{j=i}^{K}U_{j}^{i},
\end{align} 
where we let
\begin{align}
U_{j}^{i}=\sum_{l=i}^{j}{j-1 \choose l-1}t^{i}_l,\quad \forall\,j\geq i. 
\end{align}
By following similar steps as \cite[Appendix C]{gatzianas2013multiuser}, we obtain the recursive equations given by 
\begin{align}\label{eq:U^i_j}
U_{j}^{i}&=\frac{1}{1-\delta^{K-j+1}}\sum_{l=1}^{j-i}{j-1 \choose l}(-1)^{l+1}(1-\delta^{K-j+l+1})U_{j-l}^{i} 
\end{align}
for $j >i$.
Since we have $U_i^i=t^i_i=\frac{N_i}{1-\delta^{K-i+1}}$ and using the equality ${j-1 \choose c}{j-c-1 \choose i-1}={j-1 \choose j-i}{j-i \choose c}$ and the binomial theorem $\sum_{k=0}^n {n \choose k} x^k y^{n-k}=(x+y)^n$, it readily follows that we have 
\begin{align}
U_j^i&=\frac{N_i}{1-\delta^{K-j+1}}{j-1 \choose j-i}, \;\;\; j\geq i.
\end{align}
By plugging the last expression into \eqref{eq:Rggl} using \eqref{eq:e65}, we have  
\begin{align}
\tilde{R}^{i}(K)&=\frac{{K \choose i}N_i}{\sum_{j=i}^{K}\frac{N_i}{1-\delta^{K-j+1}}{j-1 \choose j-i}}\\
&=\frac{{K \choose i}}{\sum_{k=1}^{K-i+1}\frac{{K-k \choose i-1}}{1-\delta^{k}}}
%&=\frac{{K \choose i}}{\sum_{k=1}^{K-i+1}\frac{{K-k \choose i-1}}{1-\delta^{k}}},
\end{align} 
which coincides the RHS of \eqref{eq:Rorder} for $i=1,\dots, K$. This establishes the achievability proof. 
\end{proof}

As a corollary of Theorem \ref{theorem:order}, we provide an alternative expression for the sum capacity.
\begin{cor}\label{cor:Sheng}
The sum capacity of the $K$-user symmetric broadcast erasure channel with state feedback can be expressed as a 
function of $R^{2}(K), \dots, R^K(K)$ by
\begin{align}\label{eq:Sheng}
R^1(K)  = \frac{KN_1 }{ \frac{K N_1}{1-\delta^{K}} + \sum_{i=2}^K \frac{{K\choose i} N_{1\rightarrow i}}{R^i(K)}},
\end{align}
where $\frac{K N_1}{1-\delta^{K}} $ is the duration of phase 1, ${K \choose j} N_{1\rightarrow j}$ corresponds to the total number of order-$j$ packets generated in phase 1.
\end{cor}

\begin{proof} 
By letting $f$ denote the RHS of \eqref{eq:Sheng}, we wish to prove the equality $f=
R^1(K)=\frac{K}{\sum_{k=1}^K \frac{1}{1-\delta^k}}$ by proving  $f=\tilde{R}^1(K)$. If it is
true, from the achievability proof of Theorem \ref{theorem:order} that proves $\tilde{R}^i=R^i$
for all $i$, the proof is complete.  
In the RHS of \eqref{eq:Sheng}, we replace $R^i$ by the expression $\tilde{R}^i$ in \eqref{eq:Rggl} by letting $N_{1\rightarrow i}=N_i$ for $i\geq 2$. Then, we have
\begin{align}
f & =\frac{KN_1}{\frac{KN_1}{1-\delta^{K}}+\sum_{i=2}^{K}\sum_{j=i}^{K}{K \choose j}t_j^{i}}\\
&=\frac{KN_1}{\frac{KN_1}{1-\delta^{K}}+\sum_{j=2}^{K}{K \choose j}\sum_{i=2}^{j}t_j^{i}}. 
\end{align}
Comparing the desired equality $f= \tilde{R}^1(K) = \frac{KN_{1}}{\sum_{j=1}^K {K \choose j} t_j^1}$ with the above expression and noticing that 
$\frac{KN_1}{1-\delta^{K}}= Kt_1^1$, we immediately see that it remains to prove the following equality.
\begin{align}\label{eq:t_j}
t_j^1 &=\sum_{i=2}^j t^i_j ~~~~~\forall j \geq 2.
\end{align}
We prove this relation recursively. For $j=2$, the above equality follows from \eqref{eq:t^i_j} and \eqref{eq:t^i_i}.
\begin{align}
t_2^1=\frac{N_{1\rightarrow 2}}{1-\delta^{K-1}}=t_2^{2}.
\end{align}
Now suppose that \eqref{eq:t_j} holds for $l=2,\dots,j-1$ and we prove it for $j$. From \eqref{eq:t^i_j} we have
\begin{align}
t_j^1 &=\frac{1}{1-\delta^{K-j+1}}\sum_{l=1}^{j-1}{j-1 \choose l-1}N^{1}_{l\rightarrow j} \\
&=\frac{1}{1-\delta^{K-j+1}}\left[N_{1\rightarrow j} + \sum_{l=2}^{j-1}{j-1 \choose l-1} t^{1}_l \delta^{K-j+1}(1-\delta)^{j-l} \right]\label{eq:e1}\\
&=\frac{1}{1-\delta^{K-j+1}}\left[N_{1\rightarrow j} + \sum_{l=2}^{j-1}{j-1 \choose l-1} \sum_{i=2}^{l}t^{i}_{l} \delta^{K-j+1}(1-\delta)^{j-l} \right]\label{eq:e2}\\
&=\frac{1}{1-\delta^{K-j+1}}\left[N_{1\rightarrow j} + \sum_{l=2}^{j-1}{j-1 \choose l-1} \sum_{i=2}^{l} N^{i}_{l\rightarrow j}\right]\label{eq:e3}\\
&=\frac{1}{1-\delta^{K-j+1}}\left[N_{1\rightarrow j} + \sum_{i=2}^{j-1} \sum_{l=i}^{j-1} {j-1 \choose l-1}N^{i}_{l\rightarrow j}\right]\label{eq:e4}\\
&=t_j^j + \sum_{i=2}^{j-1}t^{i}_j,  %\frac{t^{1}_1 \alpha_{1\rightarrow i}}{\beta_i}\\
\end{align}
where 
\eqref{eq:e1} follows from \eqref{eq:lengthN^i}; \eqref{eq:e2}
follows from our hypothesis \eqref{eq:t_j}; \eqref{eq:e3} follows from \eqref{eq:lengthN^i}; \eqref{eq:e4} is due to the equality $\sum_{l=2}^{j-1}\sum_{i=2}^{l}=\sum_{i=2}^{j-1} \sum_{l=i}^{j-1}$; the last equality is due to \eqref{eq:t^i_j}.
Therefore, the desired equality holds also for $j$. This completes the proof of Corollary \ref{cor:Sheng}.
\end{proof}

%%%%%%%%%
\section{Achievability} \label{section:achievability2}
We provide the achievability proof of Theorem \ref{theorem:region} for the case of one-sided fair rate vector as well as the symmetric network.
The proof for the case of $K=3$ is omitted, since it is a straightforward extension of \cite[Section V]{wang2012capacity}. 

\subsection{Proposed delivery scheme for $K>3$ }\label{subsection:delivery}
We describe the proposed delivery scheme for the case of $K>3$ assuming that user $k$ requests file $W_k$ of size $F_k$ packets for $k=1,\dots, K$ without loss of generality.  
Compared to the algorithm \cite{wang2012capacity,gatzianas2013multiuser} revisited previously, our scheme must convey packets created during the  placement phase as well as all previous phases in each phase. Here is a high-level description of our proposed delivery scheme.
\begin{enumerate}
\item Placement phase (phase 0): fill the caches $Z_1,\dots, Z_K$ according to the decentralized content placement (see subsection  \ref{subsection:Decentralized}).  
This phase creates ``overheard'' packets $\{\Lc_{\Jc\setminus k}(W_k)\}$ for $\Jc\subset [K]$ and all $k$ to be delivered during phases 1 to $K$. 
\item Broadcasting phase (phase 1): the transmitter sends $V_1,\ldots,V_K$ sequentially until at least one user receives it, where $V_k=\Lc_{\emptyset}(W_k)$ 
corresponds to the order-$1$ packets.
\item Multicasting phase (phases 2-$K$): for a subset $\Jc$ of users, generate $V_{\Jc}$ as a linear combination of overheard packets during the placement phase as well as during phases $1$ to $j-1$. Send $V_{\Jc}$ sequentially for $\Jc\subseteq[K]$, 
\begin{align}%\label{eq:V_J}
V_{\Jc} = \Fc_{\Jc}\left(\{\Lc_{\Jc \setminus \Ic \cup \Ic'} (V_{\Ic})\} _{ \Ic'\Ic:\Ic' \subset \Ic \subset \Jc} , \Lc_{\Jc\setminus\{k\}}(W_{k})\right).
\end{align}
\end{enumerate}
The proposed delivery scheme achieves the optimal rate region only in two special cases. We provide the proof separately in upcoming subsections. 
%%%%%%%%%%%%%
\subsection{Proof of Theorem \ref{theorem:region} for the case of one-sided fair rate vector }\label{subsection:proof1} 
We assume without loss of generality $\delta_1\geq  \dots \geq  \delta_K$, $\delta_1R_1\geq  \dots \geq  \delta_KR_K$, and $\frac{1-p_1}{p_1} R_1\geq \dots \geq \frac{1-p_2}{p_2}R_K$. Under this setting, we wish to prove the achievability of the following equality.
\begin{align}\label{eq:desired-onesided}
\sum_{k=1}^K \frac{\prod_{j=1}^k (1-p_j)}{1-\prod_{j=1}^k \delta_j} R_k = 1. 
\end{align}
By replacing $R_k= \frac{F_{d_k}}{T_{\rm tot}}$ and further assuming $d_k =k$ for all $k$ without loss of generality, the above equality is equivalent to
\begin{align} \label{eq:desired-Ttot}
T_{\rm tot} = \sum_{k=1}^K \frac{\prod_{j=1}^k (1-p_j)}{1-\prod_{j=1}^k \delta_j} F_k.
\end{align}
The rest of the subsection is dedicated to the proof of the total transmission length
\eqref{eq:desired-Ttot}. We start by rewriting $t_{\Jc}^{\{k\}}$ in \eqref{eq:t^k_J} by
incorporating the packets generated during the placement phase. Namely we have for $k\in \Jc \subseteq [K]$  
\begin{align}\label{eq:t^k_J2}
t_{\Jc}^{\{k\}}&=\frac{\sum_{\Ic:k\in\Ic\subset \Jc}  N^{\{k\}}_{\Ic\rightarrow
\Jc}+|\Lc_{\Jc\setminus\{k\}}(W_{k})|}{1- \prod_{j\in [K]\setminus\Jc\cup\{k\} } \delta_j}. 
\end{align}
We recall that the length of sub-phase $\Jc$ is given by $t_{\Jc}=\max_{k\in\Jc}t_{\Jc}^{\{k\}}$. Our proof consists of four steps.
\paragraph{Step 1~}
We express $t_{\Jc}^{\{k\}}$ as a function of key parameters $\{\delta_k\}, \{p_k\}, \{F_k\}$ in two different ways. By following similar steps as in \cite[Appendix C]{gatzianas2013multiuser}, the aggregate length of sub-phases $\Ic \subseteq \Jc$ required by user $k$ for a fixed $\Jc\subseteq[K]$ is given by 
\begin{align}\label{eq:SumSubphase}
\sum_{\Ic:k\in\Ic\subseteq\Jc}t_{\Ic}^{\{k\}}=\frac{\prod_{j\in[K]\setminus\Jc\cup\{k\}}(1-p_j)}{1-\prod_{j\in[K]\setminus\Jc\cup\{k\}}\delta_j}F_k.
\end{align}
We have an alternative expression for $t_{\Jc}^{\{k\}}$ which is useful as will be seen shortly. 
The length of sub-phase $\Jc$ needed by user $k$ such that $k\in\Jc\subseteq[K]$ is equal to
\begin{align}\label{eq:SubphaseLength}
t_{\Jc}^{\{k\}}=\sum_{\Hc:\Hc\subseteq\Jc\setminus\{k\}}(-1)^{|\Hc|}\frac{\prod_{j\in[K]\setminus\Jc\cup\{k\}\cup\Hc}(1-p_j)}{1-\prod_{j\in[K]\setminus\Jc\cup\{k\}\cup\Hc}\delta_j}F_k. 
\end{align}
The proof is provided in Appendix \ref{appendix:subphases}.
\paragraph{Step 2~}
The length of sub-phase $\Jc$ is determined by the worst user which requires the maximum length, i.e.
$\arg\max_{k\in \Jc} t^{\{k\}}_{\Jc}$. For the special case of one-sided fair rate vector, by
using \eqref{eq:SubphaseLength} it is possible to prove that the worst user is given by 
\begin{align}\label{eq:PiOrder}
\arg\max_{k\in\Jc}{t^{\{k\}}_{\Jc}}=\min\{\Jc\}\quad,\forall\,\Jc\subseteq[K],
\end{align}
where $\min\{\Jc\}$ is the smallest index in the set of users $\Jc$ that corresponds to the user with the largest erasure probability.
The proof is provided in Appendix \ref{appendix:Ordering}.
This means that the user permutation (which determines the sub-phase length) is preserved in all sub-phases for the one-sided fair rate vector.  
 \paragraph{Step 3~} By combining the two previous steps, the total transmission length can be derived as follows. 
\begin{align}
T_{\rm tot}&= \sum_{\Jc:\Jc\subseteq[K]}\max_{k\in \Jc}t^{\{k\}}_{\Jc}\\
&= \sum_{\Jc:\Jc\subseteq[K]}t^{\{ \min\Jc\}}_{\Jc} \label{eq:st2}\\
&=\sum_{k=1}^{K}\sum_{\Jc:k\in\Jc\subseteq\{k,\dots,K\}}t^{\{ k\}}_{\Jc}\\
&=\sum_{k=1}^{K}F_k\frac{\prod_{j=1}^{k}(1-p_j)}{1-\prod_{j=1}^{k}\delta_j}, 
\end{align}
where \eqref{eq:st2} is obtained from \eqref{eq:PiOrder}; the last equality follows from \eqref{eq:SumSubphase}. Then, we obtain
the desired equality \eqref{eq:desired-Ttot}. 
%Dividing both sides by $T_{\rm tot}$ and letting $R_{_k} = \frac{F_{k}}{T_{\rm tot}}$, we readily obtain the RHS of \eqref{eq:wsr} for the identity permutation. 
\paragraph{Step 4~} The final step is to prove that under the one-sided fair rate vector \eqref{eq:desired-onesided} implies all the other $K!-1$ inequalities of the rate region \eqref{eq:wsr}. This is proved in Appendix \ref{appendix:Implication}. Hence, the achievability proof for the one-sided rate vector is completed. 

\subsection{Proof of Theorem \ref{theorem:region} for the symmetric network }\label{subsection:proof2}
First we recall the rate region of the symmetric network with uniform channel statistics and
memory sizes given in \eqref{eq:symmetric},
\begin{align} %\label
\sum_{k=1}^{K}\frac{(1-p)^k}{1-\delta^k}R_{\pi_k}\leq 1,\quad\forall \pi.
\end{align}
Exploiting the polyhedron structure and following the same footsteps as  \cite[Section V]{maddah2010degrees}, we can prove that the vertices of the above rate region are characterized as: 
\begin{align}
R_k = \begin{cases}
R_{\rm sym}( |\Kc|), k \in \Kc \\
0, k\notin \Kc 
\end{cases}
\end{align}
for $\Kc \subseteq [K]$, where the symmetric rate $R_{\rm sym}(K)$ is given by 
\begin{align}\label{eq:symmetricR}
 R_{\rm sym} (K)= \frac{1}{\sum_{k=1}^K
 \frac{\left(1-p\right)^k}{1-\delta^k}}.
\end{align}
This means that when only $|\Kc|$ users are active in the system, each of these users achieves the same symmetric rate as the reduced system of dimension $|\Kc|$. Then, it suffices to prove the achievability of the symmetric rate for a given dimension $K$. As explained in subsection \ref{subsection:delivery}, the placement phase generates ``overheard packets'' $\{\Lc_{\Jc\setminus k }(W_k)\}$ for $\Jc \subseteq [K]$ and all $k$. 
We let $N_{0\rightarrow j}=|\Lc_{\Jc\setminus k}(W_k)|$ denote the number of order-$j$ packets created during the placement phase. Then,  we can express 
the sum rate of the cached-enabled EBC by incorporating the packets generated from the placement
phase into \eqref{eq:Sheng} as follows,
\begin{align}
KR_{\rm sym}(K) = \frac{K F}{\frac{KN_{0\rightarrow 1}}{\beta_1} + \sum_{j=2}^K
\frac{{K\choose j} (N_{0\rightarrow j} + N_{1\rightarrow j})}{R^j(K)} }.
\end{align} 
By repeating the same steps as the proof of Corollary \ref{cor:Sheng}, it readily follows that the above 
expression boils down to $\frac{K}{\sum_{k=1}^K
 \frac{\left(1-p\right)^k}{1-\delta^k}}$. This establishes the achievability proof for the symmetric network. 

%%%%%%%%%
\section{Extensions} \label{section:Extensions}
In this section, we provide rather straightforward extensions of our previous results to other scenarios such as the centralized content placement and the multi-antenna broadcast channel with the state feedback. 

\subsection{Centralized content placement}
So far, we have focused on the decentralized content placement. We shall show in this subsection
that the rate region under the decentralized content placement can be easily modified to the case of the
centralized content placement proposed in \cite{maddah2013fundamental}. We restrict ourselves to
the symmetric memory size $M_k=M$ such that $M\in \{0,N/K, 2N/K, \dots, N\}$ so that the
parameter $b=\frac{MK}{N}$ is an integer. Each file is split into ${K \choose b}$ disjoint equal
size sub-files. Each sub-file is cached at a subset of users $\Jc$, $\forall\,\Jc\subseteq[K]$ with cardinality $|\Jc|=b$. Namely, the size of any sub-file of file $i$ is given by 
\begin{align}\label{eq:cent}
 |\Lc_{\Jc}(W_i)|=\frac{1}{{K \choose b}}F_i,
 \end{align}  
which satisfies the memory constraint for user $k$ 
\begin{align}
  |Z_{k}|=\sum_{i=1}^{N}\sum_{\Jc:k\in\Jc;|\Jc|=b}|\Lc_{\Jc}(W_i)|=\sum_{i=1}^{N}{K-1 \choose
  b-1}\frac{F_i}{{K \choose b}}=\sum_{i=1}^{N}\frac{b}{K}F_i=MF. \label{eq:LLN2}
\end{align}%
In analogy to Lemma \ref{lemma:decentralized} for the decentralized content placement, we can characterize the message entropy given the receiver side information. 
\begin{lemma} \label{lemma:centralized} 
For the centralized content placement \cite{maddah2013fundamental}, the following equalities hold for any $i$ and $\Kc\subseteq[K]$
\[H(W_i \cond \{Z_{k}\}_{k\in \Kc})=\frac{{K-|\Kc| \choose b}}{{K \choose b}}H(W_i).\]
\end{lemma}

\begin{proof}
Under the centralized content placement
\begin{align}
 H(W_i \cond \{Z_{k}\}_{k\in \Kc} )& = \sum_{  \Jc\subseteq [K]\setminus\Kc} H(  \Lc_{\Jc}(W_i))\\
  & = \sum_{ \Jc\subseteq [K]\setminus\Kc; |\Jc|=b} H(  \Lc_{\Jc}(W_i) ) \label{eq:tmp0893} \\
  & = \sum_{ \Jc\subseteq [K]\setminus\Kc; |\Jc|=b} \frac{1}{{K \choose b}}H(W_{i}) \label{eq:tmp0894}\\
  & =  \frac{{K-|\Kc| \choose b}}{{K \choose b}}H(W_{i}), 
\end{align}
where the first equality follows by repeating the same steps from \eqref{eq:tmp8921} to \eqref{eq:tmp892};  \eqref{eq:tmp0893} and \eqref{eq:tmp0894} follows from the definition of the centralized content placement \eqref{eq:cent}.
\end{proof}
Then, we present the rate region of the cache-enabled EBC under the centralized content placement. 
\begin{theorem} \label{theorem:region3}
For the symmetric network, the rate region of the cached-enabled EBC with the state feedback under the centralized content placement is given by
\begin{align}\label{eq:wsr3}
\sum_{k=1}^{K-b} \frac{{K-k \choose b}/{K \choose b}}{1-\delta^k} R_{\pi_k} \leq 1
\end{align}
for any permutation $\pi$ of $\{1,\dots, K\}$. 
\end{theorem}
\begin{proof}
Following the same steps as in section \ref{section:UpperBound} and replacing Lemma \ref{lemma:decentralized} with Lemma \ref{lemma:centralized}, the converse proof follows immediately. 

For achievability, as explained in subsection \ref{subsection:proof2}, it is sufficient to
consider the case of symmetric rate for a given dimension. By focusing without loss of
generality on the dimension $K$, we fix the number of packets per user to be $F$ and prove that our proposed scheme can deliver requested files to users within the total transmission length given by
\begin{align}\label{eq:LengthCentralized}
T_{\rm tot} = F\sum_{k=1}^{K-b} \frac{{K-k \choose b}/{K \choose b}}{1-\delta^k} + \epsilon_F,
\end{align} 
as $F\rightarrow \infty$. 
We proceed our proposed delivery scheme from phase $b+1$ by sending packets of order $b+1$. More precisely, in phase $b+1$ we generate and send the packets intended to
$\Jc$ by the following linear combination 
\begin{align}\label{eq:Initialization}
V_{\Jc} = \Fc_{\Jc}\left(\Lc_{\Jc\setminus k }(W_{k})\right),
\end{align}
for $\Jc \subseteq [K]$ with $|\Jc|=b+1$. In subsequent phases $b+2$ to $K$, we repeat 
\begin{align}\label{eq:Vcentralized}
V_{\Jc} = \Fc_{\Jc}\left(\{\Lc_{\Jc \setminus \Ic \cup \Ic'} (V_{\Ic})\} _{ \Ic'\Ic:\Ic' \subset \Ic \subset \Jc} \right)
\end{align}
for $\Jc \subseteq [K]$ with $|\Jc|=b+2, \dots, K$. 
%where the number of packets at each sub-phase of order $b+1$ is equal to $\frac{1}{{K \choose b}}F$.
In order to calculate the total transmission length required by our delivery algorithm, we
follow the same footsteps as in subsection \ref{subsection:proof2} and 
exploit Theorem \ref{theorem:order} on the sum capacity of order-$i$ messages that we recall here for the sake of clarity. 
\begin{align}
R^i (K) =  \frac{{K \choose i}}{\sum_{k=1}^{K-i+1} \frac{ {K-k \choose i-1}}{1-\delta^k}}.
\end{align}
Noticing that there are ${K \choose b+1}$ sub-phases in phase $b+1$ and in each sub-phase we
send a linear combination whose size is $\frac{F}{{K \choose b}}$, the total transmission length is given by 
\begin{align}
 T_{\rm tot}&=\frac{{K \choose b+1}/ {K \choose b}}{ R^{b+1}}F\\
% &=\frac{{K \choose b+1}/{K \choose b} \sum_{k=1}^{K-b} \frac{ {K-k \choose b}}{1-\delta^k}}{ {K \choose b+1}}F\\
 &=\sum_{k=1}^{K-b} \frac{ {K-k \choose b}/{K \choose b}}{1-\delta^k}F,
 \end{align}
 where the last equality follows by plugging the expression $R^{b+1}$.
% given the fact that $R=\frac{F}{T_{\rm tot}}$ we obtain
%\begin{align}
% \sum_{k=1}^{K-b} \frac{ {K-k \choose b}/{K \choose b}}{1-\delta^k}R=1
%\end{align}
%It suffices to prove the achievability for symmetrical rates since we exploit the polyhedron structure of the rate region as we did in case of decentralized cache placement. 
\end{proof}
%The following corollary holds. 
%\begin{cor}\label{cor:rate3}
%For the symmetric network($p_k=p, \delta_k=\delta~~\forall k$) and equal file size, the minimum transmission length to deliver a distinct requested file to each user in the cached-enabled EBC under centralized cache placement is given by
%\begin{align}\label{eq:rate3}
%T_{\rm tot}=\sum_{k=1}^{K}\frac{{K-k \choose b}/{K \choose b}}{1-\delta^{k}}F+\epsilon_{F}
%\end{align}
%as $F\rightarrow \infty$.
%\end{cor}
%
For the case without erasure, Theorem \ref{theorem:region3}, in particular, the expression of
the transmission length in \eqref{eq:LengthCentralized}, becomes the rate-memory tradeoff under the centralized content placement \cite{maddah2013fundamental} given by
\begin{align}
\frac{T_{\rm tot}}{F}=K\left( 1-M/N\right) \frac{1}{1+KM/N}.
\end{align}

\subsection{MISO-BC}
We consider the multi-input single-output broadcast channel (MISO-BC) between a $N_t$-antennas
transmitter and $K$ single-antenna receivers.  The channel state $S_l$ in slot $l$ is given by the $N_t\times K$ matrix and we restrict ourselves to the i.i.d. channels across time and users.  
Here, we are interested in the capacity scaling in the high signal-to-noise ratio (SNR) regime and define the degree of freedom (DoF) of user $k$ as
\[
\DoF_k=\lim_{\SNR\rightarrow\infty}\frac{R_{k}}{\log_{2} \SNR}.
\]
 We define the sum DoF of order-$j$ messages given by 
\begin{align}
\DoF^j=\lim_{\SNR\rightarrow\infty}\sum_{\Jc:|\Jc|=j}\frac{R_{\Jc}}{\log_{2} \SNR}.
\end{align}
First we recall the mains results on the MISO-BC with state feedback by Maddah-Ali and Tse
\cite{maddah2010degrees}. In \cite[Theorem 3]{maddah2010degrees},  the DoF region of the MISO-BC
with state feedback has been characterized as
\begin{align}\label{eq:DoFregion}
\sum_{k=1}^K \frac{\DoF_{\pi_k}}{k} \leq 1, ~~~ \forall \pi.
\end{align}
The sum DoF of order-$j$ messages has been characterized in \cite[Theorem 2]{maddah2010degrees} and is given by
\begin{align}\label{eq:sumDoF}
\DoF^j =\frac{{K \choose j}}{\sum_{k=1}^{K-j+1} \frac{ {K-k \choose j-1}}{k}}.
\end{align}
It is worth comparing the DoF region of the MISO-BC in \eqref{eq:DoFregion} and the capacity region of the EBC in \eqref{eq:regionEBC}. In fact, as remarked in 
 \cite{ShengISIT2015}, both regions have exactly the same structure and can be unified through a parameter $\alpha_k=k$ for the MISO-BC and $\alpha_k=1-\delta^k$ for the EBC. The same holds for the sum DoF of order-$j$ messages in the MISO-BC in \eqref{eq:sumDoF} and the sum capacity of order-$j$ packets in the EBC  
characterized in Theorem \ref{theorem:order}. 
By exploiting this duality and replacing $1-\delta^k$ with $k$ in the rate region of the symmetric EBC \eqref{eq:symmetric}, we can easily characterize
the DoF region of the cache-enabled MISO-BC with state feedback. Namely, under the decentralized content placement, the DoF region is given by 
\begin{align}
\sum_{k=1}^K \frac{(1-p)^k}{k} \DoF_{\pi_k} \leq 1, \;\; \forall \pi
\end{align}
for $N_t\geq K$,  
while under the centralized content placement, the DoF region is given by 
\begin{align}\label{eq:wsr2}
\sum_{k=1}^{K-b} \frac{{K-k \choose b}/{K \choose b}}{k} \DoF_{\pi_k} \leq 1, \;\; \forall \pi
\end{align}
for $N_t\geq K-b$.
The converse follows exactly in the same manner except that we use the entropy inequality for the MISO-BC given in \cite[Lemma 4]{ShengISIT2015} by replacing the entropy by the differential entropy and again $1-\delta^k$ by $k$. The achievability can be proved by modifying the scheme in \cite{maddah2010degrees} to the case of receiver side information along the line of \cite{piantanida2013analog}. 

As a final remark, for the case of the centralized content placement, our DoF region in \eqref{eq:wsr2} yields the following transmission length
\begin{align}
 T_{\rm tot}=\sum_{k=1}^{K-b} \frac{ {K-k \choose b}/{K \choose b}}{k}F,
 \end{align}
which coincides with \cite[Corollary 2b]{Elia2015}.
%%%%%%%%%
\section{Numerical Examples}\label{section:Examples}
\begin{figure*}
  \begin{minipage}{\columnwidth}
%\vspace{-10pt}
\begin{center}
\includegraphics[width=0.45\textwidth,clip=]{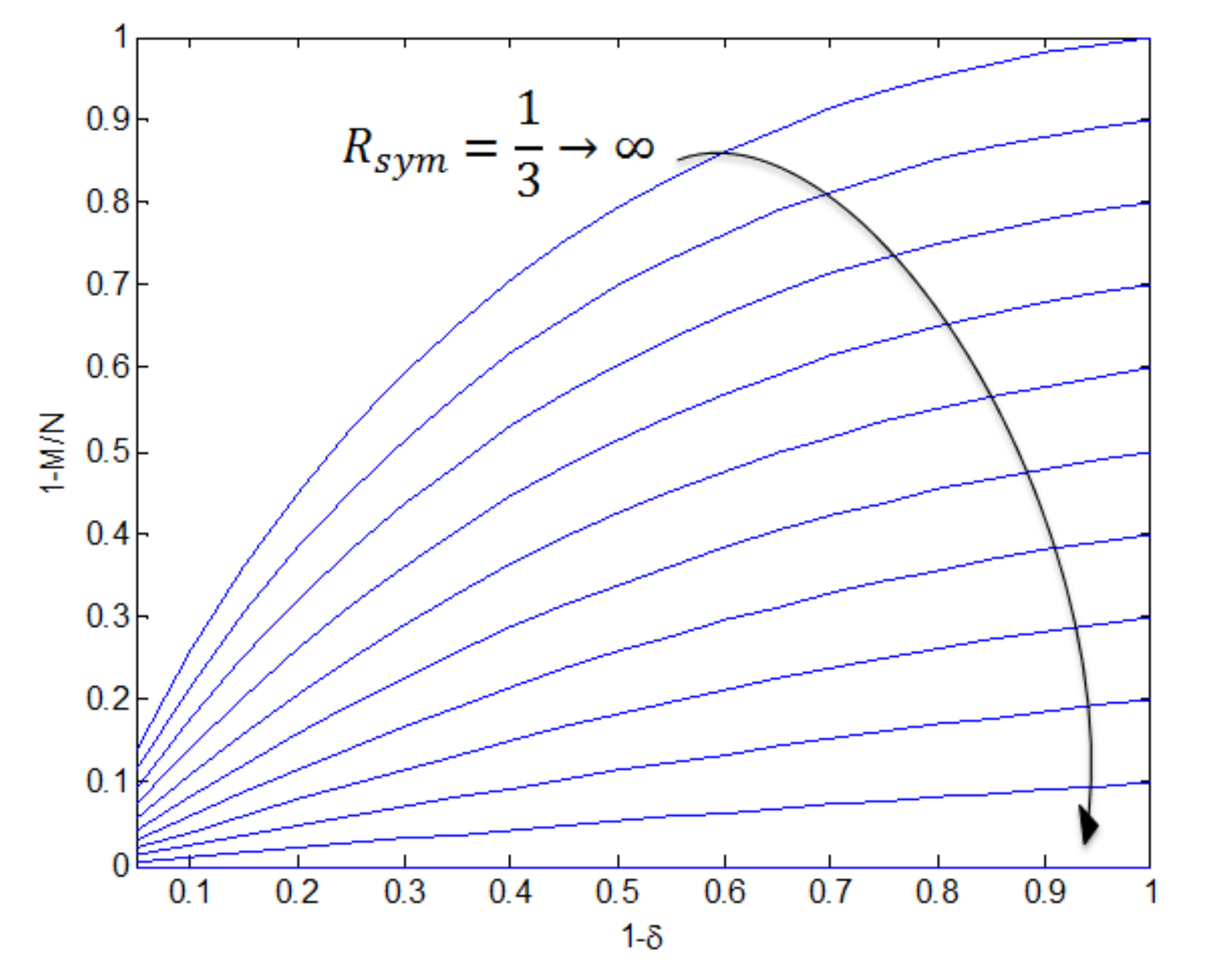}
%\vspace{-2pt}
\caption{The tradeoff between the memory and the erasure for $K=3$.}
\label{fig:memory-erasure}
\end{center}
%\vspace{-20pt}
\end{minipage}
%\begin{figure}
%\vspace{-10pt}
  \begin{minipage}{\columnwidth}
\begin{center}
\includegraphics[width=0.45\textwidth,clip=]{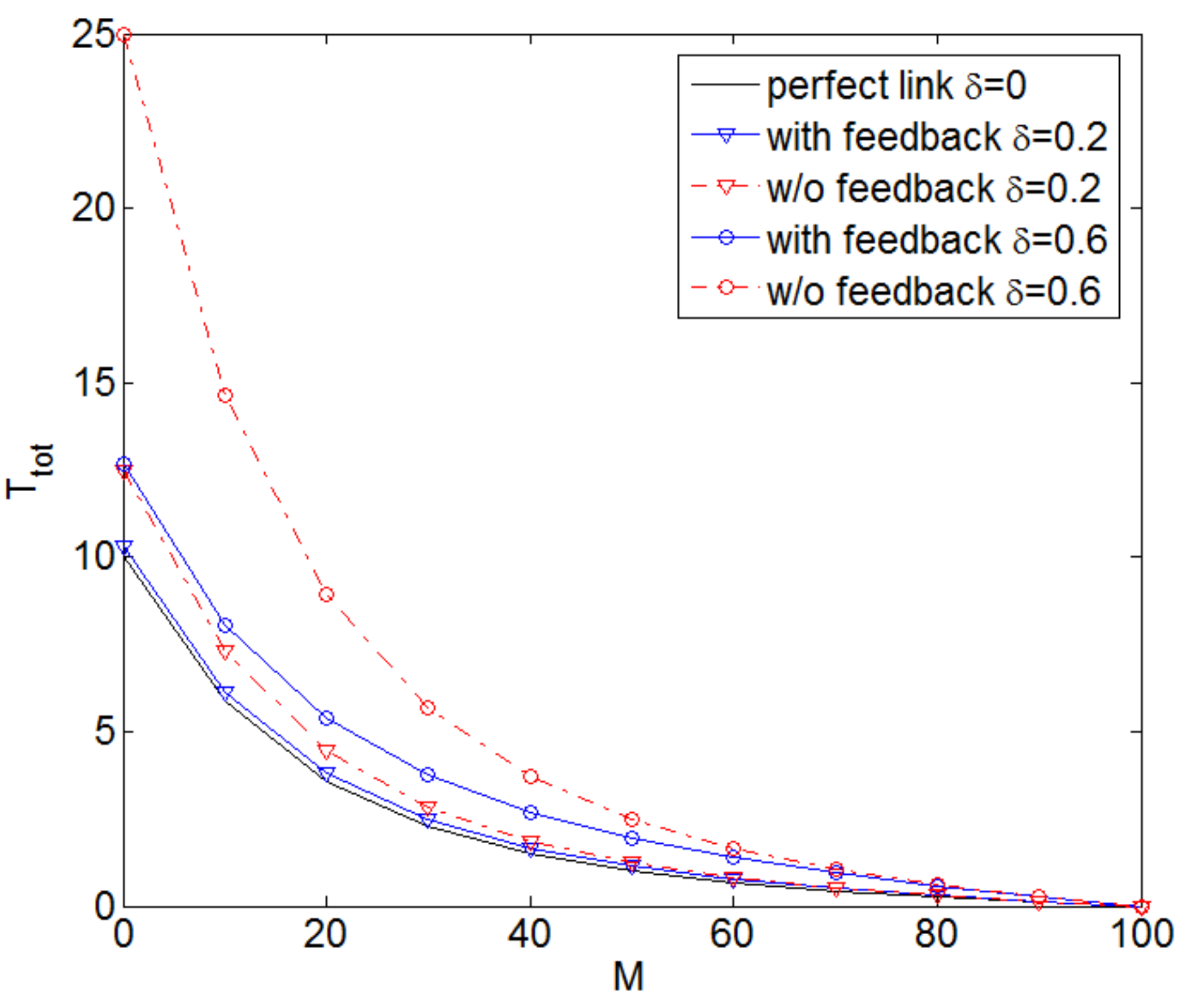}
%\vspace{-2pt}
\caption{The transmission length $T_{\rm tot}$ as a function of memory size $M$ for $N=100, K=10$. }
\label{fig:TtotvsM}
\end{center}
%\end{figure}
\end{minipage}
\end{figure*}

In this section, we provide some numerical examples to show the performance of our proposed delivery scheme. 
Fig.~\ref{fig:memory-erasure} illustrates the tradeoff between the erasure probability and the memory size for the symmetric network with $K=3$ for the case of the decentralized content placement
Each curve corresponds to a different symmetric rate $R_{\rm sym}(3)=\frac{1}{\sum_{k=1}^3\frac{(1-p)^k}{1-\delta^k}}$. The arrow shows the increasing symmetric rate from $1/3$, corresponding to case with no memory and no erasure, to infinity. The memory size increases the 
rate performance even in the presence of erasure and the benefit of
caching is significant for smaller erasure probabilities as expected 
from the analytical expression. 

Fig. \ref{fig:TtotvsM} compares the transmission length $T_{\rm tot}$, normalized by the file size $F$, achieved by our delivery scheme with feedback 
and the scheme without feedback for the case of the decentralized content placement. We consider the system with $N=100, K=10$
and the erasure probabilities of $\delta=0$ (perfect link), $0.2$, and $0.6$.  We observe that state feedback can be useful especially when the memory size is small and the erasure probability is large. In fact, it can be easily shown that the rate region of the cached-enabled EBC without feedback under the decentralized content placement is given by
\begin{align}\label{eq:NoFB-region1}
\sum_{k=1}^K \frac{\left(1-\frac{M}{N}\right)^k}{1-\delta} R_{\pi_k} \leq 1 
\end{align}
where the denominator in the LHS reflects the fact that each packet must be received by all $K$ users. This 
yields the transmission length given by
\begin{align}\label{eq:TnoFB1}
T_{\rm tot-noFB}=\frac{\sum_{k=1}^K \left(1-\frac{M}{N}\right)^k}{1-\delta}F+\epsilon_{F}. 
\end{align}
Under the centralized content placement, the rate region of the cached-enabled EBC without feedback is given by
\begin{align}\label{eq:NoFB-region2}
\sum_{k=1}^{K-b} \frac{{K-k \choose b}/{K \choose b}}{1-\delta} R_{\pi_k} \leq 1 
\end{align}
yielding 
\begin{align}\label{eq:TnoFB2}
T_{\rm tot-noFB}= \frac{K\left( 1-M/N\right) \frac{1}{1+KM/N}}{1-\delta}F+\epsilon_{F}.
\end{align}
Without state feedback, the transmission length in \eqref{eq:TnoFB1}, \eqref{eq:TnoFB2}
corresponds to the transmission length over the perfect link expanded by a factor
$\frac{1}{1-\delta}>1$, because each packet must be received by all users. The merit of feedback
becomes significant if the packets of lower-order dominate the order-$K$ packets.
The case of small $p=\frac{M}{N}$ and large erasure probability corresponds
to such a situation. 

%\begin{center}
	%\includegraphics[width=0.9\textwidth,clip=]{example4users.pdf}
	%%\vspace{-2pt}
	%%\caption{Number of transmissions vs the average memory for $\delta_i=\frac{i}{5}$, $N=20$, $K=4$ and $F_i=1$. }
	%\label{fig:TtotvsM}
%\end{center}

\begin{figure}
	\vspace{-10pt}
	\begin{center}
		\includegraphics[width=0.45\textwidth,clip=]{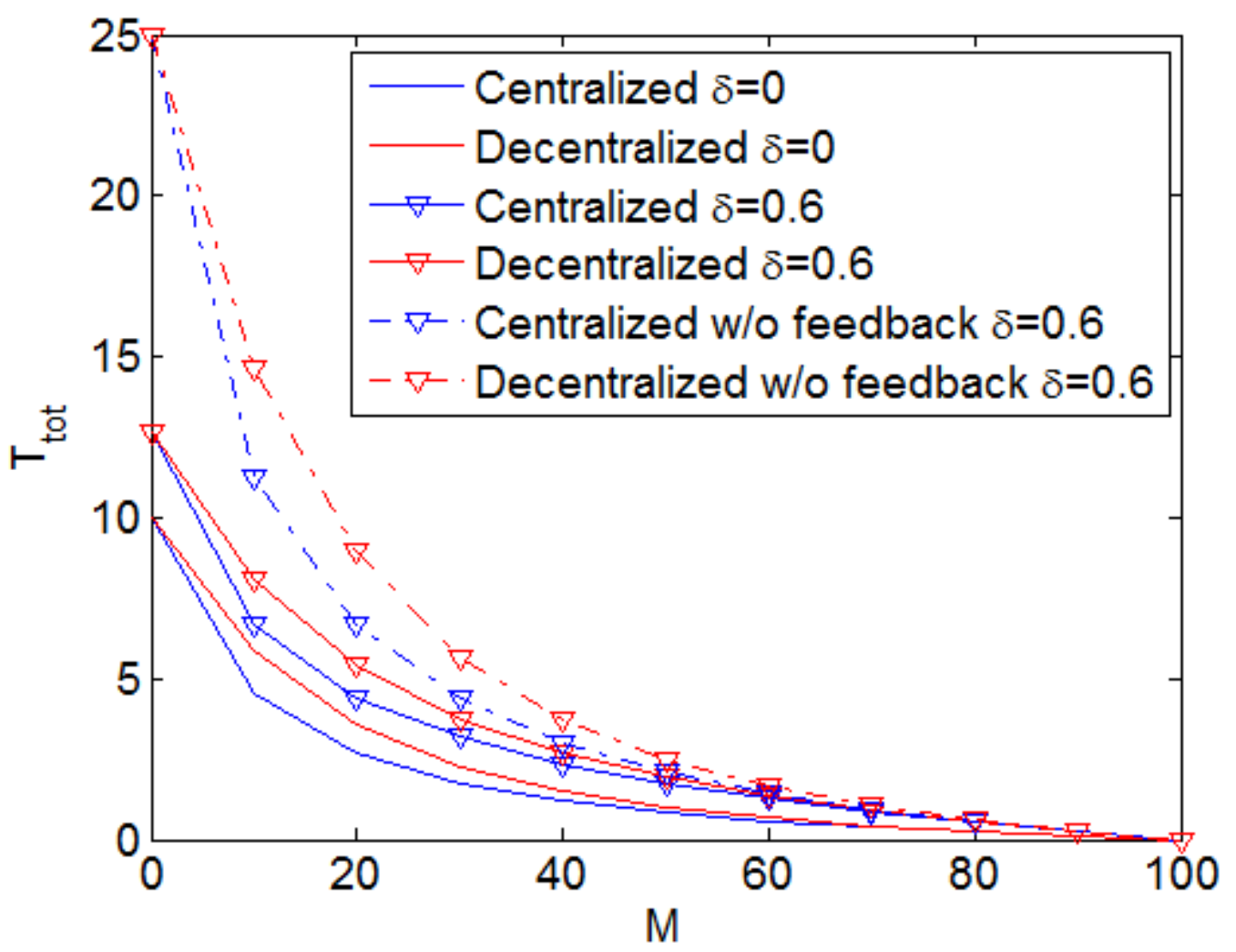}
		\vspace{-2pt}
		\caption{The transmission length $T_{\rm tot}$ as a function of memory size $M$ for $N=100, K=10$.}
		\label{fig:DECvsC}
	\end{center}
	\vspace{-15pt}
\end{figure}

\begin{figure}
	\vspace{-10pt}
	\begin{center}
		\includegraphics[width=0.45\textwidth,clip=]{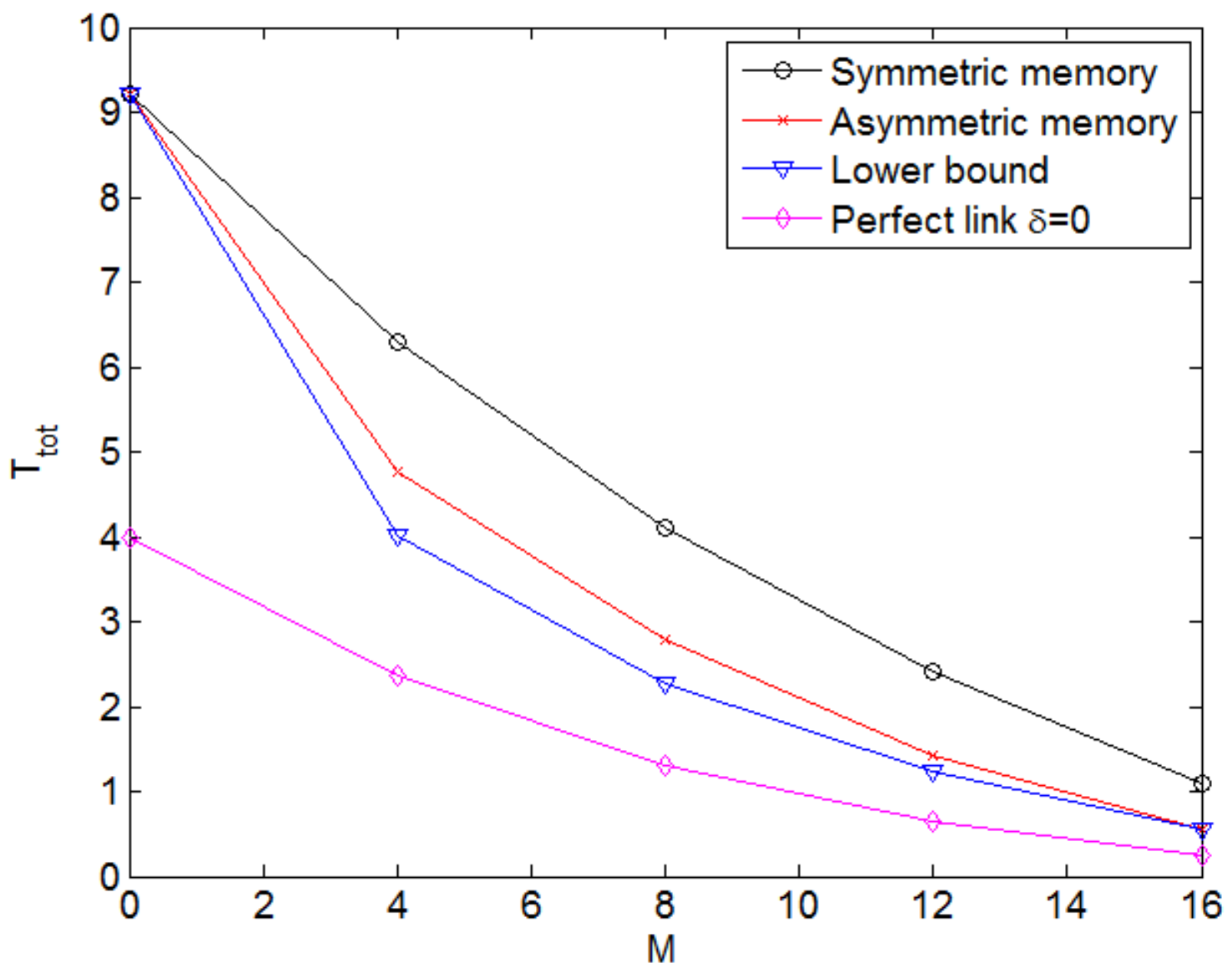}
		\vspace{-2pt}
		\caption{$\delta_i=\frac{i}{5}$, $N=20$, $K=4$ and $F_i=1$.}
		\label{fig:asym}
	\end{center}
	\vspace{-15pt}
\end{figure}

Fig. \ref{fig:DECvsC} plots the normalized transmission length $T_{\rm tot}/F$ versus the memory size $M$ in the symmetric network with $N=100, K=10$. 
 We compare the performance with and without feedback under the decentralized and the centralized caching
 for $\delta=0$ and $\delta=0.6$. %State feedback is found useful especially for small memory
 size.  The relative merit of the centralized content placement compared to the decentralized the counterpart can be observed.
  %We observe that the gap between the two caching scheme is independent from erasure,but it is larger in case of no feedback. 
 
 Fig. \ref{fig:asym} plots the normalized transmission length $T_{\rm tot}/F$ versus {\it average} memory size $M$ in the asymmetric network with $N=20$ and $K=4$ under the decentralized content placement. We let 
 erasure probabilities $\delta_k=\frac{k}{5}$ for $k=1,\dots, 4$ and consider files of equal size. We
 compare ``symmetric memory'' ($M_k=M, \forall k$), ``asymmetric memory'' obtained by optimizing
 over all possible sets of $\{M_k\}$ using our delivery scheme, as well as ``lower bound''
 obtained by optimizing over all possible of $\{M_k\}$ based on \eqref{eq:time}. This result shows the
 advantage (in terms of delivery time) of optimally allocating cache sizes across users, whenever
 possible, according to the condition of the delivery channels. 

%%%%%%%%%
\section{Conclusion}\label{section:conclusion}
In this paper, we investigated the content delivery problem in the erasure broadcast channel
(EBC) with state feedback, assuming that the content placement phase is performed with existing methods
proposed in the literature. Our main contribution was the characterization of the optimal rate
region of the channel under these conditions, based on a scheme that optimally exploits the
receiver side information acquired during the placement phase. This appears as a non-trivial
extension of the work by Wang and Gatzianas et
al.~\cite{wang2015fundamental,gatzianas2013multiuser} which have characterized the capacity
region of the EBC with state feedback for some cases of interest. We provided an intuitive
interpretation of the algorithm proposed in these works and revealed an explicit connection
between the capacity in the symmetric EBC and the degree of freedom (DoF) in the MISO-BC. More
specifically, we showed that there exists a duality in terms of the order-$j$ multicast
capacity/DoF. Such a connection was fully exploited to generalize our results to the
cache-enabled MISO-BC. Our work demonstrated the benefits of coded caching combined with state
feedback in the presence of random erasure. An interesting future direction is to include some
more practical constraints, such as the popularity profile of contents and the non-asymptotic
file size, into the current system model.

%%%%%%%%%
\appendix
In the appendix, we repeatedly use the following weight expression. 
\begin{align}\label{eq:appendix-weight}
w_{\Jc}=\frac{\prod_{j\in\Jc}(1-p_j)}{1-\prod_{j\in \Jc} \delta_j} = \frac{\overline{p}_{\Jc}}{ 1-\delta_{\Jc}}
\end{align}
where we let $\overline{p}_j= 1-p_j$ and use a short-hand notation $\delta_{\Jc}=\prod_{j\in \Jc} \delta_j$.

\subsection{Length of sub-phase} \label{appendix:subphases}
In this section, we prove \eqref{eq:SubphaseLength} given by
\begin{align}
t_{\Jc}^{\{k\}}=\sum_{\Hc:\Hc\subseteq\Jc\setminus\{k\}}(-1)^{|\Hc|}\frac{\prod_{j\in[K]\setminus\Jc\cup\{k\}\cup\Hc}(1-p_j)}{1-\prod_{j\in[K]\setminus\Jc\cup\{k\}\cup\Hc}\delta_j}F_k. 
\end{align}
To this end, we first introduce a new variable $g_{\Jc}^{\{k\}}=\frac{t^{\{k\}}_{\Jc}}{F_k}$ for $k\in \Jc \subseteq [K]$. Using \eqref{eq:SumSubphase} we obtain
\begin{align}\label{eq:SumNorSubphase}
\sum_{\Ic:k\in\Ic\subseteq\Jc}g_{\Ic}^{\{k\}}=w_{[K]\setminus\Jc\cup\{k\}}. 
\end{align}
We first need to prove the following lemma
\begin{lemma}\label{lemma:weight}
	For any nonempty set $[K]$ and $\Jc\subseteq[K]$. It holds
	\begin{align}
	\sum_{\Ic:\Ic\subseteq\Jc}\sum_{\Hc:\Hc\subseteq\Ic}(-1)^{|\Hc|}w_{[K]\setminus\Ic\cup\Hc}=w_{[K]\setminus\Jc}
	\end{align}
\end{lemma}
\begin{proof}
\begin{align}
\sum_{\Ic:\Ic\subseteq\Jc}\sum_{\Hc:\Hc\subseteq\Ic}(-1)^{|\Hc|}w_{[K]\setminus\Ic\cup\Hc}&=\sum_{\Ic:\Ic\subseteq\Jc}\sum_{\Hc:\Hc\subseteq\Ic}(-1)^{|\Hc|}w_{[K]\setminus(\Ic\setminus\Hc)}\\
&=\sum_{\Ic:\Ic\subseteq\Jc}\sum_{\Hc':\Hc'\subseteq\Ic}(-1)^{|\Ic\setminus\Hc'|}w_{[K]\setminus\Hc'}\label{eq:var1}\\
&=\sum_{\Hc':\Hc'\subseteq\Jc}\sum_{\Ic:\Hc'\subseteq\Ic\subseteq\Jc}(-1)^{|\Ic\setminus\Hc'|}w_{[K]\setminus\Hc'}\\
&=\sum_{\Hc':\Hc'\subseteq\Jc}w_{[K]\setminus\Hc'}\sum_{\Ic:\Hc'\subseteq\Ic\subseteq\Jc}(-1)^{|\Ic\setminus\Hc'|}\\
&=\sum_{\Hc':\Hc'\subseteq\Jc}w_{[K]\setminus\Hc'}\sum_{\Ic':\Ic'\subseteq\Jc\setminus\Hc'}(-1)^{|\Ic'|}\label{eq:var2}\\
&=w_{[K]\setminus\Jc}+\sum_{\Hc':\Hc'\subset\Jc}w_{[K]\setminus\Hc'}\sum_{\Ic':\Ic'\subseteq\Jc\setminus\Hc'}(-1)^{|\Ic'|}\\
&=w_{[K]\setminus\Jc}.
\end{align}
We set $\Hc'=\Ic\setminus\Hc$ and $\Ic'=\Ic\setminus\Hc'$ to obtain \eqref{eq:var1} and \eqref{eq:var2}, respectively. The last equality follows from $\sum_{\Ic:\Ic\subseteq\Jc}(-1)^{|\Ic|}=0$ for all $\Jc\neq\emptyset$.
\end{proof}
We prove \eqref{eq:SubphaseLength} by induction on $|\Jc|$. For $\Jc=\{i\}$ we have $\sum_{\Ic:i\in \Ic\subseteq \Jc}g^{\{i\}}_{\Ic}=g_{\Jc}^{\{i\}}$ and $\sum_{\Hc:\Hc\subseteq\Jc\setminus\{i\}}(-1)^{|\Hc|}w_{[K]\setminus\Jc\cup\{i\}\cup\Hc}=w_{[K]\setminus\Jc\cup\{i\}}$. By apply \eqref{eq:SumNorSubphase} for $\Jc=\{i\}$, we obtain the proof for $|\Jc|=1$.
 
%\begin{align}\label{subphaselength}
%g_{\Jc}^{\{i\}}=\sum_{\Hc\subseteq\Jc\setminus\{i\}}(-1)^{|\Hc|}w_{[K]\setminus\Jc\cup\{i\}\cup\Hc}
%\end{align}
Now suppose \eqref{eq:SubphaseLength} holds for any $\Ic\subseteq[K]$ such that $|\Ic|<|\Jc|$ and we prove in the following that it holds for $\Jc$ too.
We have
\begin{align}\label{subphaselength2}
\sum_{\Ic:i\in \Ic\subseteq \Jc}g^{\{i\}}_{\Ic}&=w_{[K]\setminus\Jc\cup\{i\}}\\
&=g_{\Jc}^{\{i\}}+\sum_{\Ic:i\in \Ic\subset \Jc}g^{\{i\}}_{\Ic}.
\end{align}
Thus, we obtain

\begin{align}
\MoveEqLeft{ g_{\Jc}^{\{i\}}=w_{[K]\setminus\Jc\cup\{i\}}-\sum_{\Ic:i\in \Ic\subset \Jc}g^{\{i\}}_{\Ic}}\\
&=w_{[K]\setminus\Jc\cup\{i\}}-\sum_{\Ic:i\in \Ic\subset \Jc}\sum_{\Hc:\Hc\subseteq\Ic\setminus\{i\}}(-1)^{|\Hc|}w_{[K]\setminus\Ic\cup\{i\}\cup\Hc}\\
&=w_{[K]\setminus\Jc\cup\{i\}}-\sum_{\Ic:i\in \Ic\subseteq \Jc}\sum_{\Hc:\Hc\subseteq\Ic\setminus\{i\}}(-1)^{|\Hc|}w_{[K]\setminus\Ic\cup\{i\}\cup\Hc}+\sum_{\Hc:\Hc\subseteq\Jc\setminus\{i\}}(-1)^{|\Hc|}w_{[K]\setminus\Jc\cup\{i\}\cup\Hc}\\
&=w_{[K]\setminus\Jc\cup\{i\}}-\sum_{\Ic:\Ic\subseteq \Jc\setminus\{i\}}\sum_{\Hc:\Hc\subseteq\Ic}(-1)^{|\Hc|}w_{[K]\setminus\Ic\cup\Hc}+\sum_{\Hc:\Hc\subseteq\Jc\setminus\{i\}}(-1)^{|\Hc|}w_{[K]\setminus\Jc\cup\{i\}\cup\Hc}\\
&=w_{[K]\setminus\Jc\cup\{i\}}-w_{[K]\setminus(\Jc\setminus\{i\})}+\sum_{\Hc:\Hc\subseteq\Jc\setminus\{i\}}(-1)^{|\Hc|}w_{[K]\setminus\Jc\cup\{i\}\cup\Hc}\label{eq:eq1}\\
&=\sum_{\Hc:\Hc\subseteq\Jc\setminus\{i\}}(-1)^{|\Hc|}w_{[K]\setminus\Jc\cup\{i\}\cup\Hc},
\end{align}
where \eqref{eq:eq1} is from Lemma \ref{lemma:weight}. 

\subsection{Existence of the permutation} \label{appendix:Ordering}
In this section, we prove that the worst user under the one-sided fair rate vector is determined by \eqref{eq:PiOrder}, namely
\begin{align}
\arg\max_{k\in\Jc}{t^{\{k\}}_{\Jc}}=\min\{\Jc\}\quad,\forall\,\Jc\subseteq[K].
\end{align}
We set $m=\min(\Jc)$ for any subset $\Jc\subseteq[K]$ such that $|\Jc|\geq2$. Proving \eqref{eq:PiOrder} is equivalent to prove  
\begin{align}\label{eq:gim}
R_m g_{\Jc}^{\{m\}}\geq R_ig_{\Jc}^{\{i\}}~~\forall i \in \Jc.
\end{align}
Recall that from our one-sided rate vector assumption we have for $i\in\Jc$ , $\delta_m\geq\delta_i$; $\delta_m R_m\geq\delta_iR_i$ and $\frac{\bar{p}_m}{p_m}R_m\geq\frac{\bar{p}_i}{p_i}R_i$. Plugging \eqref{eq:subfileWk} and \eqref{eq:Nij} into \eqref{eq:t^k_J2}, we obtain
\begin{align}\label{eq:gi}
g_{\Jc}^{\{i\}}&=\frac{1}{1-\delta_{[K]\setminus\Jc\cup\{i\}}}\left[ \sum_{\Ic:i\in\Ic\subset\Jc}g_{\Ic}^{\{i\}}\bar{\delta}_{\Jc\setminus\Ic}\delta_{[K]\setminus\Jc\cup\{i\}}+p_{\Jc\setminus\{i\}}\bar{p}_{[K]\setminus\Jc\cup\{i\}}\right], 
\end{align}
and
\begin{align}\label{eq:gm}
g_{\Jc}^{\{m\}}&=\frac{1}{1-\delta_{[K]\setminus\Jc\cup\{m\}}}\left[ \sum_{\Ic:m\in\Ic\subset\Jc}g_{\Ic}^{\{m\}}\bar{\delta}_{\Jc\setminus\Ic}\delta_{[K]\setminus\Jc\cup\{m\}}+p_{\Jc\setminus\{m\}}\bar{p}_{[K]\setminus\Jc\cup\{m\}}\right].
\end{align}
We prove by induction on $|\Jc|$ that $R_mg_{\Jc}^{\{m\}}\geq R_ig_{\Jc}^{\{i\}}$:
For $|\Jc|=2$ , $\Jc=\{m,i\}$ hence \eqref{eq:gi} and \eqref{eq:gm} imply the following
\begin{align}
g_{\Jc}^{\{i\}}&=\frac{1}{1-\delta_{[K]\setminus\Jc\cup\{i\}}}\left[ g_{i}^{\{i\}}\bar{\delta}_{m}\delta_{[K]\setminus\Jc\cup\{i\}}+p_{m}\bar{p}_{[K]\setminus\Jc\cup\{i\}}\right], 
\end{align}
and 
\begin{align}
g_{\Jc}^{\{m\}}&=\frac{1}{1-\delta_{[K]\setminus\Jc\cup\{m\}}}\left[ g_{m}^{\{m\}}\bar{\delta}_{i}\delta_{[K]\setminus\Jc\cup\{m\}}+p_{i}\bar{p}_{[K]\setminus\Jc\cup\{m\}}\right].  
\end{align}

Since $\delta_m\geq\delta_i$, it holds $\frac{1}{1-\delta_{[K]\setminus\Jc\cup\{m\}}}\geq\frac{1}{1-\delta_{[K]\setminus\Jc\cup\{i\}}}$ and $\bar{\delta}_i\geq\bar{\delta}_m$. Since $\frac{\bar{p}_m}{p_m}Rm\geq\frac{\bar{p}_i}{p_i}Ri$, then it holds $p_{i}\bar{p}_{[K]\setminus\Jc\cup\{m\}}R_m\geq p_{m}\bar{p}_{[K]\setminus\Jc\cup\{i\}}R_i$. In addition we have from \eqref{eq:SubphaseLength} : $g_m^{\{m\}}=g_i^{\{i\}}=\frac{\bar{p}_{[K]}}{1-\delta_{[K]}}$ and  $\delta_mR_m\geq\delta_iR_i$, thus we obtain $R_mg_{\Jc}^{\{m\}}\geq R_ig_{\Jc}^{\{i\}}$ for $|\Jc|=2$.

Suppose that \eqref{eq:gim} holds for any $\Ic\subseteq[K]$ such that $|\Ic|<|\Jc|$ and we prove that it holds also for $\Jc$ in the following.
%and $m=min(\Ic)$:  $R_mg_{\Ic}^{\{m\}}\geq R_ig_{\Ic}^{\{i\}}$ $\forall i\in\Ic$, and we prove it in the following that it holds also for $\Jc$.

Since $\delta_m\geq\delta_i$, it holds $\frac{1}{1-\delta_{[K]\setminus\Jc\cup\{m\}}}\geq\frac{1}{1-\delta_{[K]\setminus\Jc\cup\{i\}}}$. Since $\frac{\bar{p}_m}{p_m}R_m\geq\frac{\bar{p}_i}{p_i}R_i$, it holds\\
$p_{\Jc\setminus\{m\}}\bar{p}_{[K]\setminus\Jc\cup\{m\}}R_m\geq p_{\Jc\setminus\{i\}}\bar{p}_{[K]\setminus\Jc\cup\{i\}}R_i$. By observing \eqref{eq:gm} and \eqref{eq:gi}, it remains to prove that
\begin{align}
\MoveEqLeft{ R_m\sum_{\Ic:m\in\Ic\subset\Jc}g_{\Ic}^{\{m\}}\bar{\delta}_{\Jc\setminus\Ic}\delta_{[K]\setminus\Jc\cup\{m\}}}\geq R_i\sum_{\Ic:i\in\Ic\subset\Jc}g_{\Ic}^{\{i\}}\bar{\delta}_{\Jc\setminus\Ic}\delta_{[K]\setminus\Jc\cup\{i\}}.
 \end{align}
We have for user $m$
\begin{align}
\sum_{\Ic:m\in\Ic\subset\Jc}g_{\Ic}^{\{m\}}\bar{\delta}_{\Jc\setminus\Ic}\delta_{[K]\setminus\Jc\cup\{m\}}
&=\sum_{\Ic:\{m,i\}\subseteq\Ic\subset\Jc}g_{\Ic}^{\{m\}}\bar{\delta}_{\Jc\setminus\Ic}\delta_{[K]\setminus\Jc\cup\{m\}}
+\sum_{m\in\Ic\subset\Jc\setminus\{i\}}g_{\Ic}^{\{m\}}\bar{\delta}_{\Jc\setminus\Ic}\delta_{[K]\setminus\Jc\cup\{m\}}\\
&=\sum_{\Ic:\{m,i\}\subseteq\Ic\subset\Jc}g_{\Ic}^{\{m\}}\bar{\delta}_{\Jc\setminus\Ic}\delta_{[K]\setminus\Jc\cup\{m\}}
+\sum_{\Ic:\Ic\subset\Jc\setminus\{i,m\}}g_{\Ic\cup\{m\}}^{\{m\}}\bar{\delta}_{\Jc\setminus\Ic\setminus\{m\}}\delta_{[K]\setminus\Jc\cup\{m\}},
\end{align}
and similarly for user $i$
\begin{align}
\sum_{\Ic:i\in\Ic\subset\Jc}g_{\Ic}^{\{i\}}\bar{\delta}_{\Jc\setminus\Ic}\delta_{[K]\setminus\Jc\cup\{i\}}
&=\sum_{\Ic:\{m,i\}\subseteq\Ic\subset\Jc}g_{\Ic}^{\{i\}}\bar{\delta}_{\Jc\setminus\Ic}\delta_{[K]\setminus\Jc\cup\{i\}}+\sum_{\Ic:i\in\Ic\subset\Jc\setminus\{m\}}g_{\Ic}^{\{i\}}\bar{\delta}_{\Jc\setminus\Ic}\delta_{[K]\setminus\Jc\cup\{i\}}\\
&=\sum_{\Ic:\{m,i\}\subseteq\Ic\subset\Jc}g_{\Ic}^{\{i\}}\bar{\delta}_{\Jc\setminus\Ic}\delta_{[K]\setminus\Jc\cup\{i\}}
+\sum_{\Ic:\Ic\subset\Jc\setminus\{m,i\}}g_{\Ic\cup\{i\}}^{\{i\}}\bar{\delta}_{\Jc\setminus\Ic\setminus\{i\}}\delta_{[K]\setminus\Jc\cup\{i\}}.
\end{align}
For any $\Ic$ satisfying $\{m,i\}\subseteq\Ic\subset\Jc$ we have $|\Ic|<|\Jc|$, $min(\Ic)=m$ and $i\in\Ic$ so by the hypothesis we have $g_{\Ic}^{\{m\}}R_m\geq g_{\Ic}^{\{i\}}R_i$. In addition we have $\delta_m\geq\delta_i$ thus 
\begin{align}
\sum_{\Ic:\{m,i\}\subseteq\Ic\subset\Jc}g_{\Ic}^{\{m\}}\bar{\delta}_{\Jc\setminus\Ic}\delta_{[K]\setminus\Jc\cup\{m\}}R_m
\geq\sum_{\Ic:\{m,i\}\subseteq\Ic\subset\Jc}g_{\Ic}^{\{i\}}\bar{\delta}_{\Jc\setminus\Ic}\delta_{[K]\setminus\Jc\cup\{i\}}R_i.
\end{align}
For any $\Ic$ satisfying $\Ic\subset\Jc\setminus\{m,i\}$ we have from \eqref{eq:SubphaseLength}  $g_{\Ic\cup\{m\}}^{\{m\}}=g_{\Ic\cup\{i\}}^{\{i\}}$. In addition we have $\bar{\delta_i}\geq\bar{\delta}_m$ and $R_m\delta_m\geq R_i\delta_i$, then $\bar{\delta}_{\Jc\setminus\Ic\setminus\{m\}}\delta_{[K]\setminus\Jc\cup\{m\}}R_m\geq\bar{\delta}_{\Jc\setminus\Ic\setminus\{i\}}\delta_{[K]\setminus\Jc\cup\{i\}}R_i$. As a result we obtain 
 \begin{align}
 R_m\sum_{\Ic\subset\Jc\setminus\{i,m\}}g_{\Ic\cup\{m\}}^{\{m\}}\bar{\delta}_{\Jc\setminus\Ic\setminus\{m\}}\delta_{[K]\setminus\Jc\cup\{m\}}
 \geq R_i\sum_{\Ic\subset\Jc\setminus\{m,i\}}g_{\Ic\cup\{i\}}^{\{i\}}\bar{\delta}_{\Jc\setminus\Ic\setminus\{i\}}\delta_{[K]\setminus\Jc\cup\{i\}}.
 \end{align}
  Hence the proof is completed.

\subsection{The outer-bound under the one-sided fair rate vector} \label{appendix:Implication}

Suppose that there exists $\pi_1$ such that $\sum_{j=1}^{K}R_{\pi_1(j)}w_{\pi_1(1)..\pi_1(j)}\leq 1$ and that  $\pi_1(i)\leq\pi_1(i+1)$ holds for some  $i\in[K-1]$. We prove that for any permutation $\pi_2$ that satisfies $\pi_2(i+1)=\pi_1(i)=k$, $\pi_2(i)=\pi_1(i+1)=k'$
 and $\pi_1(j)=\pi_2(j)$ $\forall$ $j\in[K]\setminus\{i,i+1\}$, it holds $\sum_{j=1}^{K}R_{\pi_2(j)}w_{\pi_2(1)..\pi_2(j)}\leq 1$.
It suffices to show that 
\begin{align}
w_{\pi_1(1)..\pi_1(i)}R_{\pi_1(i)}+w_{\pi_1(1)..\pi_1(i+1)}R_{\pi_1(i+1)}
&\geq w_{\pi_2(1)..\pi_2(i)}R_{\pi_2(i)}+w_{\pi_2(1)..\pi_2(i+1)}R_{\pi_2(i+1)}\nonumber
\end{align}
equivalent to 
\begin{align}
 (w_{\pi_1(1)..\pi_1(i)}-w_{\pi_2(1)..\pi_2(i+1)})R_{\pi_1(i)}
 &\geq (w_{\pi_2(1)..\pi_2(i)}-w_{\pi_1(1)..\pi_1(i+1)})R_{\pi_1(i+1)}\nonumber
\end{align} 
equivalent to
\begin{align} \label{region} 
 (w_{\Ic k}-w_{\Ic kk'})R_{k}
 &\geq (w_{\Ic k'}-w_{\Ic kk'})R_{k'},
\end{align}
where $\Ic=\pi_1(1)..\pi_1(i-1)$. By replacing the weight by its expression \eqref{eq:appendix-weight} we obtain
\begin{align}
w_{\Ic k}-w_{\Ic kk'}&=\frac{\bar{p}_{\Ic k}}{1-\delta_{\Ic k}}-\frac{\bar{p}_{\Ic kk'}}{1-\delta_{\Ic kk'}} \\
 &=\bar{p}_{\Ic k}\left[ \frac{1}{1-\delta_{\Ic k}}-\frac{1}{1-\delta_{\Ic kk'}}+\frac{p_{k'}}{1-\delta_{\Ic kk'}}\right]  \\
 &=\bar{p}_{\Ic k}\left[ \frac{(1-\delta_{\Ic kk'})-(1-\delta_{\Ic k})}{(1-\delta_{\Ic k})(1-\delta_{\Ic kk'})}+\frac{p_{k'}}{1-\delta_{\Ic kk'}}\right]  \\
 &=\frac{\bar{p}_{\Ic k}}{1-\delta_{\Ic kk'}}\left[ \frac{\delta_{\Ic k}(1-\delta_{k'})}{(1-\delta_{\Ic k})}+p_{k'}\right],  
\end{align}
and similarly

\begin{align}
w_{\Ic k'}-w_{\Ic kk'}=\frac{\bar{p}_{\Ic k'}}{1-\delta_{\Ic kk'}}\left[ \frac{\delta_{\Ic k'}(1-\delta_{k})}{(1-\delta_{\Ic k'})}+p_{k}\right].  
\end{align}
Thus, \eqref{region} is equivalent to 
\begin{align} 
  \frac{\delta_{\Ic }(1-\delta_{k'})}{(1-\delta_{\Ic k})}\bar{p}_{k}\delta_{k}R_{k}-\frac{\delta_{\Ic }(1-\delta_{k})}{(1-\delta_{\Ic k'})}\bar{p}_{k'}\delta_{k'}R_{k'}+\left( \bar{p}_{k}p_{k'}R_{k}-\bar{p}_{k'}p_{k}R_{k'}\right)&\geq0. 
\end{align}
Since $k\leq k'$ then $\delta_{k}\geq\delta_{k'}$, so it is sufficient to prove that 
\begin{align} 
 \frac{\delta_{\Ic }(1-\delta_{k})}{(1-\delta_{\Ic k'})}\underbrace{\left[ \bar{p}_{k}\delta_{k}R_{k}-\bar{p}_{k'}\delta_{k'}R_{k'}\right]}_{A} +\underbrace{\left( \bar{p}_{k}p_{k'}R_{k}-\bar{p}_{k'}p_{k}R_{k'}\right)}_{B}&\geq0.
\end{align}
This is satisfied if $A\geq 0$ and $B\geq 0$. The condition B holds thanks to the definition of one-sided fair rate vector, and it  is equivalent to 
 \begin{align}
 \frac{R_{k'}}{R_k} \leq \frac{\bar{p}_{k} p_{k'}}{\bar{p}_{k'} p_k} \eqdef \theta.
 \end{align}
 
 We will examine condition A by considering the case $p_{k'}\geq p_k$ and $p_k\geq p_{k'}$ separately.
 \begin{itemize}
 \item Case $\theta >1$ \\
 In this case we have $p_k <p_{k'}$, or $\bar{p}_{k} > \bar{p}_{k'}$. 
 Condition A reduces to:
 \[   \delta_kR_k - \delta_{k'} R_{k'}  \geq 0.\]
 \item Case $\theta <1$\\
 In this case we have $p_k> p_{k'}$ or $\bar{p}_{k}< \bar{p}_{k'}$. Then we have
 \[
 \frac{R_{k'}}{R_k}  \leq \frac{\bar{p}_{k} p_{k'}}{\bar{p}_{k'} p_k }  \leq \frac{\bar{p}_{k}\delta_k}{\bar{p}_{k'}\delta_{k'} }  \leq \frac{\delta_k}{ \delta_{k'} }.
 \]
 This means that B implies A so that the desired inequality holds once B holds. Since A is inactive, we can 
 then consider a looser bounds  \[   \delta_kR_k - \delta_{k'} R_{k'}  \geq 0,\]
 which holds by the definition of one-sided fair rate vector.
 \end{itemize}
 Thus we obtain the result. Starting by $\pi_1$ as the identity we can obtain all the remaining $K!-1$ permutations.

%%%%%%%%%%%%%%%%


\begin{thebibliography}{10}
\providecommand{\url}[1]{#1}
\csname url@samestyle\endcsname
\providecommand{\newblock}{\relax}
\providecommand{\bibinfo}[2]{#2}
\providecommand{\BIBentrySTDinterwordspacing}{\spaceskip=0pt\relax}
\providecommand{\BIBentryALTinterwordstretchfactor}{4}
\providecommand{\BIBentryALTinterwordspacing}{\spaceskip=\fontdimen2\font plus
\BIBentryALTinterwordstretchfactor\fontdimen3\font minus
  \fontdimen4\font\relax}
\providecommand{\BIBforeignlanguage}[2]{{%
\expandafter\ifx\csname l@#1\endcsname\relax
\typeout{** WARNING: IEEEtran.bst: No hyphenation pattern has been}%
\typeout{** loaded for the language `#1'. Using the pattern for}%
\typeout{** the default language instead.}%
\else
\language=\csname l@#1\endcsname
\fi
#2}}
\providecommand{\BIBdecl}{\relax}
\BIBdecl


\bibitem{maddah2013fundamental}
M. Maddah-Ali and U. Niesen, ``Fundamental Limits of Caching,'' \emph{{IEEE} Trans. Inf. Theory}, vol. 60, no. 5, pp. 2856--2867, 2014.
  
\bibitem{golrezaei2011femtocaching}
N. Golrezaei, K. Shanmugam, A. G. Dimakis, A. F. Molisch, G. Caire,  ``FemtoCaching: Wireless Video Content Delivery Through Distributed Caching Helpers'' ,  \emph{{IEEE} Trans. Inf. Theory}, vol. 59, no. 12, pp. 8402--8413, 2013.

\bibitem{ji2013fundamental}
M. Ji,  G. Caire, A. Molisch,  ``Fundamental Limits of Distributed Caching in D2D Wireless Networks'' , arXiv/1304.5856, 2013.

\bibitem{ji2015order}
M. Ji,  A. Tulino, J. Llorca, and G. Caire,  ``Order-Optimal Rate of Caching and Coded Multicasting with Random Demands'', arXiv:1502.03124, 2015. 

\bibitem{maddah2013decentralized}
M. Maddah-Ali and U. Niesen, ``Decentralized Coded Caching Attains Order-Optimal Memory-Rate Tradeoff'', 
 \emph{{IEEE/ACM} Trans. on Networking}, vol. 23, no. 4, pp. 1029--1040, 2015.
%\url{http://arxiv.org/abs/1301.5848}, 2013.

\bibitem{niesen2013coded}
M. Maddah-Ali and U. Niesen, ``Coded Caching with Nonuniform Demands'',\newblock in {\em Proceedings of the IEEE Conference on Computer Communications Workshops (INFOCOM),  Toronto, Canada, 2014}, \url{http://arxiv.org/abs/1308.0178v3}, 2015.

\bibitem{pedarsani2013online}
R. Pedarsani, M. Maddah-Ali, and  U. Niesen, ``Online Coded Caching'', \url{http://arxiv.org/abs/1311.3646}, 2013.

 
\bibitem{huang2015performance}
W. Huang, S.Wang, L. Ding, F. Yang, and W. Zhang, `` The Performance Analysis of Coded Cache in Wireless Fading Channel'',
 arXiv:1504.01452v1, 2015.

\bibitem{timo2015joint}
R. Timo and M. Wigger, ``Joint Cache-Channel Coding over Erasure Broadcast Channels'', arXiv:1505.01016, 2015.

\bibitem{hachem2015effect}
J. Hachem, N. Karamchandani, and S. Diggavi, ``Effect of Number of Users in Multi-Level Coded Caching'', 
 \newblock in {\em Proceedings of the IEEE International Symposium on
 Information Theory (ISIT'2015)}, Hong-Kong, China, 2015.

\bibitem{zhangcoded}
J. Zhang, X. Lin, C. C. Wang, and X. Wang, ``Coded Caching for Files with Distinct File Sizes'', \newblock in {\em Proceedings of the IEEE International Symposium on Information Theory (ISIT'2015)}, Hong-Kong, China, 2015.


\bibitem{ShengISIT2015}
S. Yang and M. Kobayashi, ``Secrecy Communications in $K$-User Multi-Antenna Broadcast Channel with State Feedback'', \newblock in {\em Proceedings of the IEEE International Symposium on Information Theory (ISIT'2015)}, Hong-Kong, China, 2015.

\bibitem{piantanida2013analog}
P. Piantanida, M. Kobayashi, and G. Caire, ``Analog Index Coding Over Block-Fading MISO Broadcast Channels with Feedback'', 
 \newblock in {\em Proceedings of the IEEE 
Information Theory Workshop (ITW), 2013}, Seville, Spain, 2013.

\bibitem{maddah2010degrees}
M.\,A. Maddah-Ali and D.\,N.\,C.~Tse,
\newblock ``{Completely Stale Transmitter Channel State Information is Still Very Useful},''
\emph{{IEEE} Trans. Inf. Theory} vol.~58, no.~7, pp. 4418--4431, July 2012. 


\bibitem{wang2012capacity}
C. C. Wang, \newblock ``{On the Capacity of 1-to-Broadcast Packet Erasure Channels with Channel Output feedback}'',  
\emph{{IEEE} Trans. Inf. Theory}, vol. 58, no. 2, pp. 931--956, February 2012.

\bibitem{gatzianas2013multiuser}
M. Gatzianas, L. Georgiadis,  and L. Tassiulas,  \newblock ``{Multiuser Broadcast Erasure Channel With Feedback-Capacity and Algorithms}'', \emph{{IEEE} Trans. Inf. Theory}, vol.~59, no.~9, pp. 5779--5804, September 2013. 

\bibitem{wang2015fundamental}
 S. Wang, W. Li, X. Tian, and H. Liu, ``{Coded Caching with Heterogenous Cache Sizes}'', arXiv preprint arXiv:1504.01123v3, 2015.


\bibitem{allerton2015}
 A. Ghorbel, M. Kobayashi, and S. Yang, ``{Cache-Enabled Broadcast Packet Erasure Channels with State Feedback}'', 
 \newblock in {\em Proceedings of the 53rd Annual Allerton Conference on Communication, Control, and Computing (Allerton)}, IL, USA,  2015.

 \bibitem{Ajaykrishnan}
 N. Ajaykrishnan, N. Prem, S. Prabhakaran, M. V., and R. Vaze, ``{Critical Database Size for Effective Caching}'',arXiv preprint arXiv:1501.02549,  2015.
 
\bibitem{Allerton2014}
S. Karthikeyan, M. Ji, A. Tulino, J. Llorca and A. Dimakis ``{Finite Length Analysis of Caching-Aided Coded Multicasting}'', 
 \newblock in {\em Proceedings of the 52nd Annual Allerton Conference on Communication, Control, and Computing (Allerton)}, IL, USA, 2014

\bibitem{Ji2013}
 M. Ji, G. Caire, and A. Molisch. ``{Fundamental Limits of Distributed Caching in D2D Wireless Networks}'', 
 \newblock in {\em Proceedings of the IEEE Information Theory Workshop (ITW), 2013}, Seville, Spain, 2013.

\bibitem{Hierarchical2014}
N. Karamchandani, U. Niesen, M. Maddah-Ali, and S. Diggavi "Hierarchical Coded Caching." ,
\newblock in {\em Proceedings of the IEEE International Symposium on Information Theory (ISIT'2014)}, HI, USA, 2014.
 
 \bibitem{Hachem2014}
J. Hachem, N. Karamchandani, and S. Diggavi, ``{Coded Caching for Heterogeneous Wireless Networks with Multi-Level Access}'' arXiv preprint arXiv:1404.6560v2, 2015.

\bibitem{el2011network}
A. El Gamal and Y. H. Kim, ``{Network Information Theory}'', Cambridge university press, 2011.


 \bibitem{zhang2015coded}
J. Zhang, F. Engelmann, and P. Elia ``Coded Caching for Reducing CSIT-Feedback in Wireless Communications'', \newblock in {\em Proceedings of the 53rd Annual Allerton Conference on Communication, Control, and Computing (Allerton)}, IL, USA,  2015.

 \bibitem{Elia2015}
J. Zhang, and P. Elia, ``{Fundamental Limits of Cache-Aided Wireless BC: Interplay of Coded-Caching and CSIT Feedback}'', arXiv preprint arXiv:1511.03961, 2015. 
 
 \end{thebibliography}
\end{document}